\documentclass[12pt]{elsarticle}

\usepackage{graphicx}
\usepackage{amssymb}

\usepackage[utf8]{inputenc}
\usepackage[english]{babel}

\usepackage{amsmath}
\usepackage{graphicx}
\usepackage[colorinlistoftodos]{todonotes}
\usepackage{hyperref}
\usepackage{amsthm}
 
\usepackage{textcomp}
\usepackage{ amssymb }

 \usepackage{mathdots}
\usepackage{latexsym}
\usepackage{mathtools}
\usepackage{lipsum}
\usepackage{graphicx}
\usepackage{float}
\usepackage{wrapfig}
\restylefloat{figure}
\usepackage{caption}
\usepackage{subcaption}
\usepackage{gensymb}

\newtheorem{definition}{Definition}[section]

\newtheorem{proposition}[definition]{Proposition}

\newtheorem{lemma}[definition]{Lemma}

\newtheorem{corollary}[definition]{Corollary}

\begin{document}

\begin{frontmatter}


\title{Probabilistic Pursuits on Graphs}



\author[add1]{Michael Amir}

\ead{ammicha3@cs.technion.ac.il}

\author[add1]{Alfred M. Bruckstein}

\ead{freddy@cs.technion.ac.il}

\address[add1]{Technion Israel Institute of Technology, Haifa, Israel}

\journal{Theoretical Computer Science}

\begin{abstract}
We consider discrete dynamical systems of ``ant-like'' agents engaged in a sequence of pursuits on a graph environment. The agents emerge one by one at equal time intervals from a source vertex $s$ and pursue each other by greedily attempting to close the distance to their immediate predecessor, the agent that emerged just before them from $s$, until they arrive at the destination point $t$. Such pursuits have been investigated before in the continuous setting and in discrete time when the underlying environment is a regular grid. In both these settings the agents' walks provably converge to a shortest path from $s$ to $t$. Furthermore, assuming a certain natural probability distribution over the move choices of the agents on the grid (in case there are multiple shortest paths between an agent and its predecessor), the walks converge to the uniform distribution over all shortest paths from $s$ to $t$ - and so the agents' locations are on average very close to the straight line from $s$ to $t$.

In this work we study the evolution of agent walks over a general finite graph environment $G$. Our model is a natural generalization of the pursuit rule proposed for the case of the grid. The main results are as follows. We show that ``convergence'' to the shortest paths in the sense of previous work extends to all pseudo-modular graphs (i.e. graphs in which every three pairwise intersecting disks have a nonempty intersection), and also to environments obtained by taking graph products, generalizing previous results in two different ways. We show that convergence to the shortest paths is also obtained by chordal graphs (i.e. graphs in which all cycles of four or more vertices have a chord), and discuss some further positive and negative results for planar graphs. In the most general case, convergence to the shortest paths is not guaranteed, and the agents may get stuck on sets of recurrent, non-optimal walks from $s$ to $t$. However, we show that the limiting distributions of the agents' walks will always be uniform distributions over some set of walks of equal length.

\end{abstract}


\end{frontmatter}

\section{Introduction}
\label{chap:intro}


The notion of a greedy chain pursuit of agents was inspired by a trail of ants headed from their nest to sources of food. The question of whether ants need to estimate geometrical properties of the underlying surface to converge to the optimal path was posed by Feynman \cite{feynman} and has since been investigated in both robotics and biology, inspiring research into networks, cooperative multi-agent algorithms and dynamical systems (cf. \cite{ants1, ants2, ants3, ants4, ants5}). Here we are interested in whether, when ant-like agents pursue each other using a simplistic logic that requires no persistent states and only local information regarding the environment, the ``trails'' they traverse on the graph converge to the shortest paths from the source vertex $s$ to the destination vertex $t$. We are also interested in the probability distribution of these trails as time goes to infinity.

Various notions of pursuit have been extensively investigated for graphs under the subject of ``cops and robber'' games, where one or more agents attempt to capture a moving target (see \cite{coprobber} for a general survey). A greedy pursuit rule  was investigated in e.g. \cite{coprobber2}. Our objective is completely different, as we are instead interested in the structure over time of the trails formed by configurations of agents in pursuit of each other. 

The model of probabilistic chain pursuit discussed in this work was first introduced by \cite{ants2} for the case of the grid graph. In this model, a sequence of agents $A_0, A_1, A_2 \ldots$ emerges from a source vertex $s$ at times $0, \Delta, 2\Delta, \ldots$ for a fixed integer $\Delta > 1$. The agent $A_0$ moves in an arbitrary way, until, in a finite number of steps, it reaches the destination vertex $t$ and subsequently stops there forever. For any $i > 0$, the agent $A_i$ chases $A_{i-1}$ until they both arrive and stop at $t$.

Denote the position of $A_i$ at time $T$ as $(x_i,y_i)$ and let $d_x = |x_i - x_{i-1}|$, $d_y = |y_i - y_{i-1}|$. At every time step, $A_i$ may take a single step on either the $x$ or the $y$-axis of the grid but not both. Every move that it makes must bring it closer to the current location of $A_{i-1}$ (unless they both stand in the same place, in which case $A_i$ does not move). Most of the time, $A_i$ will be able to get closer to $A_{i-1}$'s location via either the $x$-axis or the $y$-axis. To account for this, $A_i$ chooses, according to a probabilistic rule, whether it will move on the $x$ or the $y$-axis. Specifically, $A_i$ will move on the $x$ axis with probability $\frac{d_x}{d_x + d_y}$ and on the $y$ axis with probability $\frac{d_y}{d_x + d_y}$. 

A corner case that needs to be mentioned is when $d_x + d_y = 0$. This can only occur when $A_i$ lies on the same vertex as $A_{i+1}$. In \cite{ants2} such situations are handled by merging $A_i$ and $A_{i+1}$ to single agent. In this work we instead assume that $A_i$ is allowed to stay put until $A_{i+1}$ distances itself from it. This difference is insignificant, and all of our results are applicable to either case (with only minor technical modifications to the proofs).

When agents chase each other according to this pursuit rule, the walk of the agent $A_i$ converges to the optimal path from $s$ to $t$ as $i$ tends to infinity, irrespective of the initial path of $A_0$. Furthermore, the limiting distribution of the walks of the agents is the uniform distribution over all shortest paths from $s$ to $t$. Remarkably, since in a grid graph (drawn on the plane in the usual way) the vast majority of shortest paths from $s$ to $t$ pass through vertices which are close to the straight line from $s$ to $t$, the positions of the agents $A_0, A_1, A_2, \ldots$ will lie very close to this line almost all the time, and so the ``ant trails'' that the agents form on the grid will almost always look approximately like a straight line.

The purpose of this work is to study an extension of the model proposed in \cite{ants2} wherein pursuit takes place on a fixed but arbitrary graph $G$.

Previously we have mentioned our assumption that the agent $A_i$ is allowed to ``stay put'' when it is on the same vertex as $A_{i-1}$. Formally, we enable it to do this by always assuming that the graph $G$ is reflexive (i.e. the edge $v \to v$ exists for every vertex $v$). We consider the following pursuit rule:

\begin{definition}[Pursuit rule]
\label{Pursuitrule}
The agent $A_i$, for all $i > 0$, pursues agent $A_{i-1}$ by selecting, uniformly at random, a shortest path in $G$ from its current vertex to the current location of $A_{i-1}$, then moves to the next vertex of that path. (The shortest path from a vertex $v$ to itself is defined as the reflexive path $v\to v$. Hence when both $A_i$ and $A_{i-1}$ are located on the same vertex, $A_i$ will stay in place).
\end{definition}

When the underlying graph is a grid, this rule coincides with the pursuit rule of \cite{ants2} described above.
 
\subsection*{Overview}

The structure of this work is as follows. Section 2 formalizes our terminology and model, and establishes a preliminary characterization of what we call ``convergent graphs'' and ``stable graphs'', that we use in the rest of this work. Section 3 contains the results about convergence and stability of various classes of graphs. Section 4 shows that all possible distributions of agent walks occuring at the limit are uniform distributions over some set of walks. All sections use results stated in Section 2, but Sections 3 and 4 are largely independent.

All graphs in which pursuit takes place are assumed to be finite, bidirectional and reflexive, with at most one edge per pair of vertices $(v,u)$. A walk is a sequence of vertices $v_1 v_2\ldots v_n$ such that an edge exists between each $v_i, v_{i+1}$. A path is a walk with no repeating vertices. The length of a path or walk is the number of vertices it traverses, and is written $|P|$. For example, if $P = v_1v_2v_3$ then $|P| = 3$.

\section{Preliminary Characterizations}

We are given an undirected finite graph $G$ and two vertices $s, t \in G$ (not necessarily different). The vertex $s$ will be called the source vertex, and $t$ the destination vertex. Ants (a(ge)nts) emerge at $s$ and pursue the ants that left before them as they walk towards the destination vertex (where they shall stop). Specifically, the first ant, $A_0$, follows an arbitrary finite walk from $s$ to $t$, taking one step across an edge of this walk every unit of time. Furthermore, every $\Delta$ units of time - for some parameter $\Delta$ - the ant $A_i$ leaves $s$ and pursues ant $A_{i-1}$ (taking one step per time unit) according to the pursuit rule described in Definition \ref{Pursuitrule}.

Here, $\Delta$ (``delay time'') is an important parameter, and we will always assume that it is greater than 1. The pursuit rule guarantees that the distance between $A_i$ and $A_{i-1}$ is bounded by $\Delta$. This is because the distance is initially at most $\Delta$, and at every time step $A_{i}$ will close the distance to $A_{i-1}$'s location by one vertex and $A_{i-1}$ will move away from that location by at most one vertex, so the distance between them (so long as it is above $0$) is either preserved or shortened. By the same reasoning, whenever there is a point in the pursuit where $d(A_i, A_{i-1}) = x > 0$, the distance between the agents will never subsequently increase above $x$. The one exception is when $d(A_i, A_{i-1}) = 0$ and both agents haven't yet stopped at $t$. In such cases, the distance might increase to $1$ after one time step (as $A_{i}$ will stay in the same spot but $A_{i-1}$ might step to another vertex), but it will remain at most $1$ from there on (note that since $\Delta > 1$ this does not contradict the earlier bound on the distance). If the distance $d(A_i, A_{i-1})$ is $x$ once $A_{i-1}$ reaches $t$, then $A_i$ will arrive at $t$ after $x$ subsequent steps. Since $A_{i-1}$ arrived at $s$ precisely $\Delta$ time steps before $A_i$, this means that $A_i$ will arrive at $t$ in $\Delta - x \geq 0$ less steps than $A_{i-1}$. Consequently, the walk lengths of subsequent agents in the chain pursuit are non-increasing. 

Note that in order to carry out the pursuit rule, every ant needs only local information about the graph (specifically, it need only know the disk of radius $\Delta$ about its current vertex).

We wish to understand the lengths of the walks of the ant $A_i$ as $i \to \infty$, and their eventual distribution, assuming an arbitrary initial walk for ant $A_0$. In particular, we want to know if, in finite expected time, the walk of $A_i$ will be an optimal path from $s$ to $t$.

We denote the graph walk taken by $A_i$ as $P(A_i)$. To model the distribution of ant walks over time we consider a Markov chain $\mathcal{M}_{\Delta} (s,t)$, parametrized by the vertices $s$, $t$ and a positive integer $\Delta$, whose states are all (the infinitely many) possible walks from $s$ to $t$. The transition probability from $P_1$ to $P_2$ is defined as $Prob[P(A_{i+1})=P_2 | P(A_i) = P_1]$, assuming the given value of $\Delta$. 

In this work we are primarily interested in studying the closed communicating classes of $\mathcal{M}_{\Delta} (s,t)$ (closed classes for short). A closed communicating class is a subset of states of $\mathcal{M}_{\Delta} (s,t)$ such that every two states communicate with each other, and no state in the set communicates with a state outside the set. A state $P_i$ is said to communicate with state $P_j$ if it is possible, with probability greater than 0, for the chain to transition from $P_i$ to $P_j$ in a finite number of steps (the reader may find more detail in the first chapter of \cite{markovref}, or in \cite{markovref2}).

Consider the (finitely many) walks and closed communicating classes of $\mathcal{M}_{\Delta} (s,t)$ reachable from the initial state $P(A_0)$. Once an ant takes a walk that belongs to a closed communicating class, the walk choices of all ants that emerge in the future will belong to this class. Additionally, since the length of $P(A_0)$ is finite and ants cannot transition to a walk of length greater than this, the walks of our ants become such that they belong to a closed communicating class in finite expected time (see \cite{markovref}). Hence, closed communicating classes capture the notion of ``stabilization'' in our dynamical system.

\setcounter{figure}{0}
\begin{figure}[!ht]
\caption{The behavior over time of chain pursuit where the graph environment is a grid with two ``holes''.}
\centering
\begin{subfigure}{.5\textwidth}
\centering
\includegraphics[width=\linewidth]{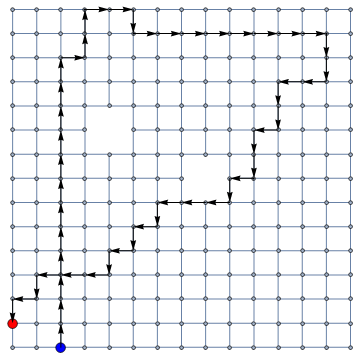}
\end{subfigure}%
\begin{subfigure}{.5\textwidth}
\centering
\includegraphics[width=\textwidth]{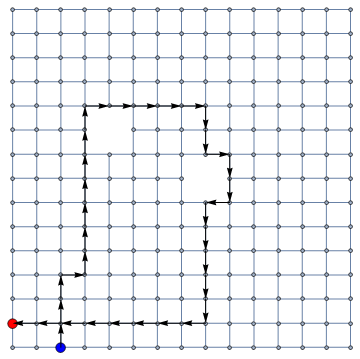}
\end{subfigure}
\begin{subfigure}{.5\textwidth}
\centering
\includegraphics[width=\linewidth]{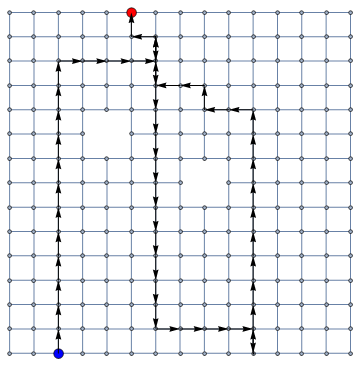}
\end{subfigure}%
\begin{subfigure}{.5\textwidth}
\centering
\includegraphics[width=\linewidth]{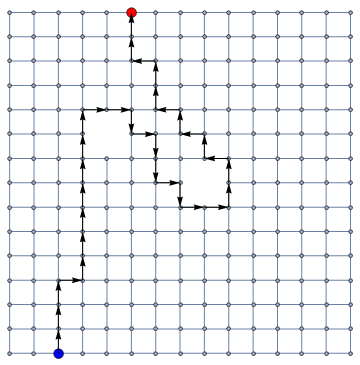}
\end{subfigure}
\label{fig:graphsimulation}
\end{figure}

In Figure \ref{fig:graphsimulation} we show the behavior over time of our dynamical system, for $\Delta = 2$, over a grid with ``holes''. Each image illustrates the walk taken by some $A_i$ (here we show the walk of every agent individually, though in an actual simulation, multiple agents would be walking on the graph concurrently). The image to the left shows the walk taken by $A_0$, and the image to the right shows a walk belonging to the closed communicating class at which the pursuit stabilized. As we can see, though the agents typically manage to shorten the walk of the initial agent $P(A_0)$ over time, they are by no means guaranteed to converge to an optimal path. 

In fact, for the particular graph environment of Figure \ref{fig:graphsimulation} (a grid with ``holes''), there exist an infinite number of closed communicating classes, containing walks of arbitrary length. The loops of the walk $P(A_0)$ around each of the holes determines the kinds of walks the ants may converge to.

Nevertheless, closed communicating classes $\mathcal{C}$ have some nice regularity properties. As an initial observation we will show that the walks belonging to a closed class must all have the same length, i.e. traverse the same number of vertices, denoted by $|P|$.

\begin{lemma}
\label{uniformlength}
Let $\mathcal{C}$ be the set of walks in a closed class of $\mathcal{M}_{\Delta} (s,t)$. Then for all $P_l, P_k \in \mathcal{C}$, $|P_l| = |P_k|$. 
\end{lemma}

\begin{proof}
Assume that there are two walks in \textit{C} such that $|P_l| < |P_k|$. The walk $P_l$ is a reachable state belonging to the closed class. However, once an ant takes walk $P_l$, it will never take $P_k$ again since by the pursuit rule, walk lengths are monotonically non-increasing. Thus $P_k$ is not recurrent in the class, a contradiction. 
\end{proof}

Of particular interest are closed classes that contain the shortest paths from $s$ to $t$, or a subset of these paths. In \cite{ants2} it was shown that when $G$ is a grid graph, $\mathcal{M}_{\Delta} (s,t)$ has a unique closed class, that contains all the shortest paths. Later we will see that whenever $\mathcal{M}$ has a unique closed class, it is necessarily the class of all shortest paths.

A concept that will be used in our arguments is \textit{$\delta $-optimality}:

\begin{definition} We will say that a walk $P = v_1 v_2  ...  v_n$ is $\delta$-optimal, if, for any two vertices $v_i, v_{i+\delta} \in P$ we have that $d_{G}(v_i, v_{i+\delta}) = \delta$ (where $d_G(u,v)$ denotes the distance between two vertices in $G$).
\end{definition}

We note that $d_{G}(v_i, v_{i+\delta}) = \delta$ implies that $v_i v_{i+1} \ldots v_{i+\delta}$ is a shortest path from $v_i$ to $v_{i+\delta}$. Therefore we have by extension that $d_{G}(v_i, v_{i+k}) = k$ for any $k \leq \delta$.

A simple but important observation is that walks which belong to a closed class $\mathcal{C}$ of $\mathcal{M}_{\Delta} (s,t)$ must be $\Delta$-optimal:

\begin{lemma}
\label{Toptimallemma}
Let $\mathcal{C}$ be the set of walks of some closed class of $\mathcal{M}_{\Delta} (s,t)$. Then any walk in $\mathcal{C}$ is $\Delta$-optimal.
\end{lemma}

\begin{proof}
Suppose for contradiction that there is some walk $P \in \mathcal{C}$ that is not $\Delta$-optimal. Since we are in a closed class, this walk is recurrent. Thus after finite expected time some ant $A_i$ will take walk $P$. After $\Delta$ time, an ant $A_{i+1}$ will start chasing $A_i$ via some path in $\mathcal{C}$, maintaining a distance of $\Delta$ or less from $A_i$ at every time step. 

If $A_{i+1}$ ever manages to decrease the distance to $A_i$ below $\Delta$, then it will arrive at $t$ in less steps than $A_i$, contradicting Lemma \ref{uniformlength}. Hence, we will assume that the initial distance between $A_{i+1}$ and $A_i$ is $\Delta$. We shall show that there exists a legal ``pursuit strategy'' for ant $A_{i+1}$ that can occur with non-zero probability, and will cause the distance between $A_i$ and $A_{i+1}$ to drop below $\Delta$. This leads to a contradiction (due to Lemma \ref{uniformlength}).

Since $A_i$ follows a non-$\Delta$-optimal walk, there exists a vertex $v_l$ along $P$, such that $dist(v_l, v_{l+\Delta}) < \Delta$. Assume that $l$ is the minimal index for which this occurs. Since $l$ is minimal, by the pursuit rule, $A_{i+1}$ will with some probability pursue $A_i$ using the vertices $v_1, v_2 \ldots v_{l}$.

Once ant $A_{i+1}$ arrives at $v_l$, $A_i$ will be at $v_{l+\Delta}$, having walked from $v_l$ to $v_{l+\Delta}$ along the vertices of $P(A_i)=P$. Consequently, at this point in time, we will have that $d(A_i, A_{i+1}) = d(v_l, v_{l+\Delta}) < \Delta$. Since $A_{i+1}$ has successfully dropped the distance between itself and $A_i$ below $\Delta$, which was the distance between them at the moment $A_{i+1}$ emerged from $s$, the walk $P(A_{i+1})$ must be shorter than $P(A_i)$. This directly contradicts Lemma \ref{uniformlength}.
\end{proof}

We will be primarily interested in graphs $G$ for which closed classes contain only shortest paths from $s$ to $t$; in other words, graphs on which convergence to a shortest path is guaranteed. We will consider as a special case graphs for which there is a unique closed class which contains \textit{all} shortest paths. 

\theoremstyle{definition}
\begin{definition}[Convergent graphs] Let $\mathcal{M}_{\Delta} (s,t)$ be a Markov chain defined over the graph $G$. If all closed classes in $\mathcal{M}_{\Delta} (s,t)$ contain only shortest paths from $s$ to $t$, then $G$ is called $(s,t)$-\textbf{convergent} with respect to delay time $\Delta$. When this holds for \textit{all} pairs of vertices $(s,t)$, and any $\Delta > 1$, $G$ is called \textit{convergent}.
\end{definition}

\theoremstyle{definition}
\begin{definition}[Stable graphs] If, for a fixed $\Delta$, $\mathcal{M}_{\Delta} (s,t)$ has a unique closed class, then $G$ is called $(s,t)$-\textbf{stable} with respect to delay time $\Delta$. When this holds for \textit{all} pairs of vertices $(s,t)$, and any $\Delta > 1$, the graph $G$ is called \textit{stable}.
\end{definition}

In every Markov chain $\mathcal{M}_{\Delta} (s,t)$ over any graph $G$ there is at least one closed class that contains a shortest path (since walk lengths are monotonically non-increasing, and since in $G$ there is a shortest path from $s$ to $t$ and we can set that path to be $P(A_0)$). Thus if $G$ is stable, the unique closed class of $\mathcal{M}_{\Delta} (s,t)$ contains only shortest paths from $s$ to $t$. Thus a stable graph is in particular a convergent graph:

\begin{proposition}
If $G$ is $(s,t)$-stable with respect to $\Delta$, then it is $(s,t)$-convergent with respect to $\Delta$. 
\end{proposition}

In fact, as will be shown, stable graphs are precisely the graphs for which the walks of the ants converge to a unique limiting distribution over \textit{all} shortest paths from $s$ to $t$.

An example of a graph which is convergent but not stable is given in Figure 2, (a). It can be proven simply and directly that it is convergent, with the tools developed later in this section. It is not stable, since if $A_0$ takes the top shortest path from $s$ to $t$ no subsequent $A_i$ will ever use the bottom shortest path.\newline\newline

\setcounter{figure}{1}
\usetikzlibrary{automata}
\begin{figure}[h]
\caption{Examples of graphs which are not convergent or not stable.}
\centering
\begin{subfigure}{.8\textwidth}
\centering
\label{convergentnotstablefigure}
\usetikzlibrary{shapes.geometric}
\begin{tikzpicture}
[every node/.style={inner sep=0pt}]
[every node/.minimum size=1cm]
\node[label={[xshift=-0.3cm, yshift=-0.25cm]s}] (1) [circle,fill] at (87.5pt, -150.0pt) {1};
\node (2) [circle,fill] at (112.5pt, -137.5pt) {1};
\node (3) [circle,fill] at (137.5pt, -137.5pt) {1};
\node (4) [circle,fill] at (112.5pt, -162.5pt) {1};
\node (5) [circle,fill] at (137.5pt, -162.5pt) {1};
\node (6) [circle,fill] at (162.5pt, -137.5pt) {1};
\node (7) [circle,fill] at (162.5pt, -162.5pt) {1};
\node (8) [circle,fill] at (187.5pt, -137.5pt) {1};
\node (9) [circle,fill] at (187.5pt, -162.5pt) {1};
\node[label={[xshift=0.3cm, yshift=-0.25cm]t}] (10) [circle,fill] at (212.5pt, -150.0pt) {1};
\draw (1) to  (2);
\draw (1) to  (4);
\draw (4) to  (5);
\draw (2) to  (3);
\draw (4) to  (2);
\draw (5) to  (3);
\draw (3) to  (6);
\draw (5) to  (7);
\draw (8) to  (6);
\draw (7) to  (9);
\draw (8) to  (10);
\draw (9) to  (10);
\draw (7) to  (6);
\draw (9) to  (8);
\end{tikzpicture}
\subcaption{A convergent graph that is not a stable graph.}
\end{subfigure}
\newline\newline\newline\newline
\begin{subfigure}{.8\textwidth}
\centering
\begin{tikzpicture}
[every node/.style={minimum size=8pt, inner sep=0pt}]
\draw[black] (0:1) arc (0:360:10mm);
\foreach \phi in {1,...,4}{
	\node[circle,fill] at (360/5 * \phi:1cm) {\small \textcolor{black}1};
      }
\node[label={[xshift=0.27cm, yshift=-0.3cm]s}][circle,fill] at (360/5 * 5:1cm)  {\textcolor{black}{\small{1}}};     
\node[label={[xshift=0.25cm, yshift=-0.20cm]t}][circle,fill] at (360/5 * 6:1cm)  {\textcolor{black}{\small{t}}};   
\end{tikzpicture}
\subcaption{$C_5$: A graph which is not convergent.}
\end{subfigure}
\label{fig:examples1}
\end{figure}
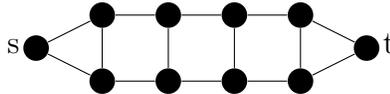
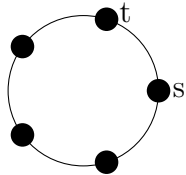

The simplest example of a graph that is not convergent (so also not stable) is the cycle on $n$ vertices $C_n$ for $n \geq 5$ (Figure 2, b). In fact, this graph has an infinite amount of closed classes for Markov chains with $\Delta = 2$ formed by looping clockwise from $s$ to $t$ and back to $s$ any number of times, where $s$ and $t$ are selected to be two adjacent vertices. It is easy to come up with additional graphs where convergence fails (these are typically, but not necessarily, sparse graphs with big cycles), and some more involved examples are provided in section 2.3.

The choice of $\Delta$ is important when discussing $(s,t)$-stability or convergence (but not for the generic properties ``stable graph'' and ``convergent graph''). For instance, the 5-cycle in Figure 2, (b) is $(s,t)$-convergent with respect to delay time $3$, but not delay time $2$.

Our first characterization of stable and convergent graphs will be in terms of local ``deformations'' of walks to one another.

\theoremstyle{definition}
\begin{definition}[$\delta$-deformability] 
\begin{enumerate} 
\item Let $P=U_1 U_2 U_3$ and $P^* = U_1 U_2' U_3$ be two walks from $v_1$ to $v_n$, such that $U_i$ is a (possibly empty) sub-walk in $G$, and such that $|U_2'| \leq |U_2| \leq \delta-1$. Then $P^*$ is said to be an \textit{atomic} $\delta$-deformation of $P$. 
\item If there is a sequence of atomic $\delta$-deformations of $P$ that results in $P'$, then $P'$ is said to be a $\delta$-deformation of $P$. 
\end{enumerate}
\end{definition}

A $\delta $-deformation of $P$ can result in a walk of shorter length (by replacing $U_2$ with a shorter or even empty sub-walk), but will never lengthen it. Atomic 2-deformations are illustrated in Figure \ref{fig:examples2}. Another example can be seen in Figure 8, (b), later on. \newline\newline

\setcounter{figure}{2}
\begin{figure}[h]
\centering
\begin{tikzpicture}[scale=1.6]
[every node/.style={minimum size=10pt, inner sep=0pt}]
\node (1) at (112.5pt, -187.5pt) {\small \textcolor{black}{$v_1$}};
\node (2) at (112.5pt, -150.0pt) {\small \textcolor{black}{$v_2$}};
\node (3) at (150.0pt, -150.0pt) {\small \textcolor{black}{$v_3$}};
\node (4) at (150.0pt, -187.5pt) {\small \textcolor{black}{$v_2'$}};
\node (5) at (212.5pt, -187.5pt) {\small \textcolor{black}{$v_1$}};
\node (6) at (237.5pt, -162.5pt) {\small \textcolor{black}{$v_2$}};
\node (7) at (262.5pt, -187.5pt) {\small \textcolor{black}{$v_3$}};
\draw [->, color=black] (2) to  (3);
\draw [->, dashed, color=black] (4) to  (3);
\draw [->, dashed, color=black] (1) to  (4);
\draw [->, color=black] (1) to  (2);
\draw [->, color=black] (5) to  (6);
\draw [->, color=black] (6) to  (7);
\draw [->, dashed, color=black] (5) to  (7);
\end{tikzpicture}
\caption{Illustration of a 2-deformation with and without walk shortening. Illustrated are walks drawn on the 4-cycle and 3-cycle graphs. On the left, we deform the path $v_1 v_2 v_3$ to $v_1 v_2' v_3$. On the right, we deform $v_1v_2v_3$ to $v_1v_3$.\newline}
\label{fig:examples2}
\end{figure}
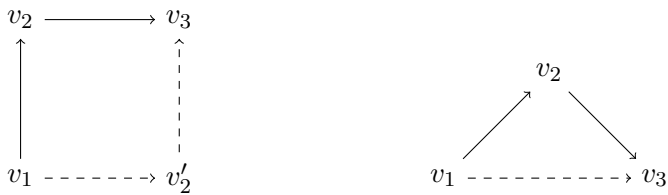

Intuitively, an atomic $\delta $-deformation of a walk $P$ is a \textit{small local change} (replacing at most $\delta-1$ vertices) of $P$. We prove the following:

\begin{proposition}
\label{convergentdeform}
$G$ is $(s,t)$-convergent with respect to $\Delta$ if and only if every walk from $s$ to $t$ is $\Delta$-deformable to a shortest path from $s$ to $t$.
\end{proposition}

\begin{proposition}
\label{stabledeform}
$G$ is $(s,t)$-stable if and only if there exists a fixed shortest path $P$ from $s$ to $t$ such that every walk from $s$ to $t$ is $\Delta$-deformable to $P$. 
\end{proposition}

Proposition \ref{convergentdeform} can be interpreted as saying that if any walk $P$ is transformable to a shortest path by a sequence of small, local $\Delta$-changes to its vertices, then our ants are guaranteed to find a shortest path to their destination (in finite expected time), and vice-versa.

For the proofs, we will use the following lemmas.

\begin{lemma}
\label{closedwalk}
If $P$ is a walk belonging to a closed communicating class of $\mathcal{M}_{\Delta} (s,t)$, then so is any $\Delta$-deformation of $P$.
\end{lemma}

\begin{proof}
First we prove the statement for the case of atomic $\Delta$-deformations. Let $P' = U_1U_2'U_3$ be an atomic $\Delta$-deformation of $P=U_1U_2U_3$. Suppose $A_i$ takes walk $P$. We will show $A_{i+1}$ can take walk P' with non-zero probability.

We note that since $P$ is $\Delta$-optimal (due to belonging to a closed class), we must have precisely $|U_2'| = |U_2|$ (otherwise the $\Delta$-optimality of $P$ would break at the last vertex of $U_1$). Therefore we also have that $|P| = |P'|$.  

Let $P = u_1 u_2 \ldots u_n$ and $P' = u_1 \ldots u_j u_{j+1}' \ldots u_{j+\Delta-1}' u_{j+\Delta} \ldots u_n$, such that $U_2' = u_{j+1}' \ldots u_{j+\Delta-1}'$. It suffices to show that if $A_{i+1}$ is standing on the $k$th vertex of $P'$ and $A_i$ on the $(k+\Delta)$th vertex of $P$, then $A_{i+1}$ may move to the $(k+1)$th vertex of $P'$ in the next time step in pursuit of $A_i$. We show this by separation to cases:

\begin{enumerate}
\item If $A_{i+1}$ is standing on $u_k$, for $1 \leq k < j$, then $A_{i}$ is standing on $u_{k+\Delta}$, and by the $\Delta$-optimality of $P$ we have that $d(u_k, u_{k+\Delta}) = \Delta$, implying also that $d(u_{k+1}, u_{k+\Delta}) = \Delta - 1$. Hence, $A_{i+1}$ may with some probability move to $u_{k+1}$, due to the pursuit rule.

\item If $A_{i+1}$ is standing on $u_j$, then since $|U_2'| = |U_2|$ we have that $d(u_{j+1}', u_{j+\Delta}) \leq \Delta - 1$ (in fact this is an equality). Furthermore, similar to (1) we have that $d(u_{j}, u_{j+\Delta}) = \Delta$, so we may have $A_{i+1}$ move to $u_{j+1}'$ in the next time step.

\item If $A_{i+1}$ is standing on $u_k'$ for $j < k < j+\Delta$, then since $|U_2| = |U_2'|$ we have that $d(u_k', u_{k+\Delta}) \leq \Delta$. However, since $P$ belongs to a closed communicating class, this inequality cannot be strict (otherwise $A_{i+1}$ would've been able to lower its overall walk length below $|P(A_i)|$). Therefore we have that $d(u_k', u_{k+\Delta}) = \Delta$. Furthermore due to the indices and the fact that $P'$ is a walk we have that $d(u_{k+1}', u_{k+\Delta}) \leq \Delta-1$, so $A_{i+1}$ may move to $u_{k+1}'$ (or $u_{k+1}$, if $k = \Delta - 1$) in pursuit of $A_i$. 

\item If $j + \Delta \leq k$, then the proof proceeds similar to (1).
\end{enumerate}

This shows that $A_{i+1}$ may always, with some probability, pursue $A_i$ by taking the next vertex of $P'$, and this in turn shows that $P'$ belongs to the same closed class as $P$. Since any $\Delta$-deformation of $P$ can be constructed from a sequence of atomic deformations, the proof is complete.
\end{proof}

\begin{lemma}
\label{walkdeform}
For a $\Delta$-delay chain pursuit, we have that: for all $i$, $P(A_{i+1})$ is a $\Delta$-deformation of $P(A_i)$.
\end{lemma}

\begin{proof}

We will define a sequence $P_0, P_1, P_2, \ldots P_N$ of atomic $\Delta$-deformations that deform $P(A_i)$ to $P(A_{i+1})$, such that $P_0 = P(A_i)$ and $P_N = P(A_{i+1})$. $P_k$ is defined recursively, based on $P_{k-1}$.

Write $P_0 = v_1 \ldots v_{n_1}$ and $P_N = u_1 \ldots u_{n_2}$, where $v_1 = u_1$ and $v_{n_1} = u_{n_2}$. Note that by the pursuit rule we have for all $r$ that $d(u_r, v_{r+\Delta}) \leq \Delta$ and that $u_{r+1}$ lies on a shortest path from $u_r$ to $v_{r+\Delta}$. Thus there is always a shortest path from $u_r$ to $v_{r+\Delta}$ passing through at most $\Delta - 1$ vertices, starting with $u_{r+1}$. 

For a given $k$, $P_k$ is a $\Delta$-deformation of $P_{k-1}$, replacing the sub-walk $u_k \ldots v_{k+\Delta}$ with a shortest path from $u_k$ to $v_{k+\Delta}$ passing through $u_{k+1}$ (if $k+\Delta > n$ we set $v_{k+\Delta} = v_n$).

An example of the sequence for $\Delta = 3$ and $|P_0| = |P_N| = 6$ is seen below. $x_i$ is an arbitrary vertex along a shortest path, as constrained by the definition of $P_k$.

\large
\begin{center}
	$\textrm{ {} {} {} {} {} {} } \mathrm{v_1v_2v_3v_4v_5v_6} \stackrel{3-def.}{\rightsquigarrow} $
    
	$\textrm{ {} {} {} }\mathrm{v_1\mathbf{u}_2\mathbf{x}_3v_4v_5v_6} \rightsquigarrow $
    
    $\textrm{ {} {} {} }\mathrm{v_1u_2\mathbf{u}_3\mathbf{x}_4v_5v_6} \rightsquigarrow $
    
	$\textrm{ {} {} {} }\mathrm{v_1u_2u_3\mathbf{u}_4\mathbf{x}_5v_6} \rightsquigarrow $
    
	$\mathrm{u_1u_2u_3u_4\mathbf{u}_5u_6} \notag $
\end{center}
\normalsize

The following is always true:

\begin{enumerate}
\item $P_k$ is a valid walk in $G$.

\item $P_k$ is always a $\Delta$-deformation of $P_{k-1}$, as it replaces at most $\Delta - 1$ vertices (note that $P_1$ is well-defined since $u_1 = v_1$)

\item The first $k$ vertices of $P_k$ match those of $P_N$.
\end{enumerate}

So we see that for some $N < n_2$, we will have $P_N = P(A_{i+1})$, completing the proof.
\end{proof}

An important corollary of the Lemmas just proven is that a closed communicating class $\mathcal{C}$ is equal, precisely, to the closure under $\Delta$-deformations of any walk $P \in \mathcal{C}$. 

To prove Proposition \ref{convergentdeform} we apply the lemmas. In one direction, assume that the graph is convergent. Then for every walk $P$ there is a valid sequence of walks taken by successive ants (the first ant taking walk $P$) that goes to a shortest path. By Lemma \ref{walkdeform} this means that every walk can be $\Delta$-deformed to a shortest path. The other direction follows from lemma \ref{closedwalk}, and due to the fact that the walks of the ants will enter some closed class of $\mathcal{M}$ in finite expected time. \qed

To prove Proposition \ref{stabledeform}, in one direction, suppose that the graph is stable. So it has a unique closed class. Recall that every stable graph is a convergent graph. Due to Proposition \ref{convergentdeform} and \ref{walkdeform}, this implies that every walk is $\Delta$-deformable to a shortest path from the unique closed class. All paths in this class are deformable to each other due to the mentioned closure, so it follows that every walk is deformable to a fixed (shortest) path. In the other direction, suppose every walk is deformable to the same shortest path $P$. Then it immediately follows from \ref{closedwalk} that there is just one closed class, so the graph is stable. \qed

Earlier we mentioned that the unique closed class of $\mathcal{M}_{\Delta} (s,t)$ when $G$ is stable necessarily contains \textit{all} shortest paths from $s$ to $t$. We can now prove this. To start, we know the unique closed class, that we will denote $\mathcal{C}$, can contain only shortest paths. Now let $P$ be a shortest path not in $\mathcal{C}$. Since the successors of any agent following this path must eventually (in expected finite time) end up taking a path in $\mathcal{C}$, it follows from Lemma \ref{walkdeform} that there is a sequence of atomic $\Delta$-deformations of $P$ - each deforming it necessarily to another shortest path - such that at the end of this sequence is a shortest path belonging to $\mathcal{C}$. Let $P'$ be the penultimate path in this sequence, that is, the last one not belonging to $\mathcal{C}$. Any $\Delta$-deformation $P^*$ of $P'$ cannot shorten it (as it is a shortest path), thus it necessarily replaces a sub-walk of length $k \leq \Delta-1$ in $P'$ with a different sub-walk of length $k$. But by doing the reverse we can $\Delta$-deform $P^*$ back to $P'$, thus $P'$ belongs to $\mathcal{C}$ - contradiction. This gives a stronger characterization of stable graphs:

\begin{proposition}
\label{stabledeformpaths}
$G$ is $(s,t)$-stable (with respect to $\Delta$), iff its unique closed class contains all shortest paths from $s$ to $t$. 
\end{proposition}

$\mathcal{M}_{\Delta} (s,t)$ has a unique closed class, and it is easy to see that it is aperiodic when restricted to this class (since any shortest path taken by $A_i$ has a chance of repeating itself for $A_{i+1}$). Thus it has a unique limiting distribution. Proposition \ref{stabledeformpaths} implies that $\mathcal{M}_{\Delta} (s,t)$ has a unique limiting distribution if and only if it has a unique limiting distribution over all shortest paths from $s$ to $t$. In Section 4 we will show that this distribution must in fact be the uniform distribution over all shortest paths from $s$ to $t$. 

In some sense the delay time $\Delta =2$ is a benchmark for whether a graph is convergent (resp. stable):

\begin{proposition} 
\label{convergentT2}
$G$ is convergent if and only if it is $(s,t)$-convergent for any pair of vertices $(s,t)$, with respect to $\Delta =2$.
\end{proposition}

\begin{proposition} 
\label{stableT2}
$G$ is stable if and only if it is $(s,t)$-stable for any pair of vertices $(s,t)$, with respect to $\Delta =2$.
\end{proposition}

Both these propositions follow immediately from the fact that any $2$-deformation is in particular a $\Delta$-deformation for all $\Delta > 2$. Thus any graph which is convergent (resp. stable) with respect to $\Delta = 2$ is also convergent (resp. stable) for $\Delta > 2$. \qed 

In other words, we need only prove that a graph is convergent (stable) for $\Delta =2$ to show that it is convergent (stable) for any $\Delta >1$. The significance of this is that we can now consider stability and convergence as properties of graphs rather than properties of the dynamical system defined by our pursuit rule. As a consequence of the statements we have proved in this section, we can forget the pursuit rule and consider only the relations between walks that exist in $G$. Hence in the following sections, \textbf{we will assume that $\Delta =2$}, unless stated otherwise. 

\newpage 

\section{Convergent and Stable Graphs}

Having set up some helpful propositions in the previous section, we can begin to discuss several interesting classifications of stable and convergent graphs. In \cite{ants2} it was shown that the grid is stable. We focus on generalizing this result to broad classes of graphs that include the grid as a special case. 

We want to study graphs whose  $\delta$-deformation can easily be understood. We find it fruitful to study graphs whose induced distance metric has certain kinds of constraints placed on it (intuitively, constraints that relate to the idea of a simply connected or convex space or to operations that preserve these). To this end we show that (i) all pseudo-modular graphs are stable, and that (ii) the stable- and convergent-graph properties are preserved under taking graph products. We then move to a discussion of  chordal and planar graphs. 

We state the following definition and lemma, which will become useful in several sections:

\begin{definition}
\label{discrepdefine}
Let $P = v_1 \ldots v_n$ be a walk in $G$. We call two vertices $v_i$, $v_j$ \textit{discrepancy vertices} of $P$, if they minimize the difference of indexes $j-i$ under the constraint that $j-i > d(v_i, v_j)$.
\end{definition}

If $P$ is the shortest path from $v_1$ to $v_n$, then $P$ has no discrepancy vertices. On the other hand any non-optimal walk must have at least one pair of such vertices. We note that due to $\Delta$-optimality, if $P$ belongs to a closed communicating class of $\mathcal{M}_{\Delta} (s,t)$, then for any pair of discrepancy vertices $v_i$, $v_j$, $j-i$ must be larger than $\Delta$; in particular, under our assumption that $\Delta =2$ we always have $j-i > 2$.

\begin{lemma}
\label{discrepobserve}
Let $v_i$, $v_j$ be discrepancy vertices in some walk $P = v_1 v_2 \ldots v_n$. Then:
\begin{enumerate}
\item There is no vertex $v$ in the sub-walk $v_{i+1} v_{i+2} \ldots v_{j-1}$ such that $d(v_i,v) + d(v,v_j) = d(v_i,v_j)$. (That is, no vertex between $v_i$ and $v_j$ in $P$ belongs on a shortest path between them, or is equal to one of them).
\item $j - i \leq d(v_i, v_j) + 2$
\end{enumerate}
\end{lemma}

\begin{proof}
For proof of (1), we note that had there been such a vertex, say $v_t \in P$, then either $v_i, v_t$ or $v_t, v_j$ would've been discrepancy vertices instead of $v_i, v_j$ since the difference of indexes is smaller.

For (2), assume that $j-i > d(v_i, v_j) + 2$. From this we have: $j-(i+1) > d(v_i,v_j) + 1 \geq d(v_{i+1},v_j)$, thus $v_{i+1},v_j$ are discrepancy vertices - a contradiction to the minimality property of discrepancy vertices, since $j-(i+1) < j-i$. 
\end{proof}

\subsection{Pseudo-modular Graphs}

\setcounter{figure}{3}
\begin{figure}[h]
\caption{Pseudo-modular graphs.}
\centering
\begin{subfigure}{.5\textwidth}
\centering
\scalebox{0.8}{
\begin{tikzpicture}
[every node/.style={inner sep=0pt}]
\node (1) [circle,fill] at (125.0pt, -137.5pt) {\textcolor{black}{'}};
\node (2) [circle,fill] at (100.0pt, -137.5pt) {\textcolor{black}{'}};
\node (3) [circle,fill] at (100.0pt, -162.5pt) {\textcolor{black}{'}};
\node (4) [circle,fill] at (125.0pt, -162.5pt) {\textcolor{black}{'}};
\node (5) [circle,fill] at (100.0pt, -112.5pt) {\textcolor{black}{'}};
\node (6) [circle,fill] at (125.0pt, -112.5pt) {\textcolor{black}{'}};
\node (7) [circle,fill] at (150.0pt, -112.5pt) {\textcolor{black}{'}};
\node (8) [circle,fill] at (150.0pt, -137.5pt) {\textcolor{black}{'}};
\node (9) [circle,fill] at (100.0pt, -187.5pt) {\textcolor{black}{'}};
\node (10) [circle,fill] at (125.0pt, -187.5pt) {\textcolor{black}{'}};
\node (11) [circle,fill] at (100.0pt, -87.5pt) {\textcolor{black}{'}};
\node (12) [circle,fill] at (125.0pt, -87.5pt) {\textcolor{black}{'}};
\node (13) [circle,fill] at (150.0pt, -87.5pt) {\textcolor{black}{'}};
\node (14) [circle,fill] at (175.0pt, -87.5pt) {\textcolor{black}{'}};
\node (15) [circle,fill] at (175.0pt, -112.5pt) {\textcolor{black}{'}};
\node (16) [circle,fill] at (175.0pt, -137.5pt) {\textcolor{black}{'}};
\node (17) [circle,fill] at (75.0pt, -187.5pt) {\textcolor{black}{'}};
\node (18) [circle,fill] at (75.0pt, -162.5pt) {\textcolor{black}{'}};
\node (19) [circle,fill] at (75.0pt, -137.5pt) {\textcolor{black}{'}};
\node (20) [circle,fill] at (75.0pt, -112.5pt) {\textcolor{black}{'}};
\node (21) [circle,fill] at (75.0pt, -87.5pt) {\textcolor{black}{'}};
\draw [line width=0.625, color=black] (17) to  (18);
\draw [line width=0.625, color=black] (18) to  (19);
\draw [line width=0.625, color=black] (19) to  (20);
\draw [line width=0.625, color=black] (20) to  (21);
\draw [line width=0.625, color=black] (21) to  (11);
\draw [line width=0.625, color=black] (11) to  (5);
\draw [line width=0.625, color=black] (5) to  (2);
\draw [line width=0.625, color=black] (2) to  (3);
\draw [line width=0.625, color=black] (3) to  (9);
\draw [line width=0.625, color=black] (20) to  (5);
\draw [line width=0.625, color=black] (19) to  (2);
\draw [line width=0.625, color=black] (18) to  (3);
\draw [line width=0.625, color=black] (17) to  (9);
\draw [line width=0.625, color=black] (3) to  (4);
\draw [line width=0.625, color=black] (9) to  (10);
\draw [line width=0.625, color=black] (2) to  (1);
\draw [line width=0.625, color=black] (1) to  (4);
\draw [line width=0.625, color=black] (10) to  (4);
\draw [line width=0.625, color=black] (1) to  (6);
\draw [line width=0.625, color=black] (5) to  (6);
\draw [line width=0.625, color=black] (11) to  (12);
\draw [line width=0.625, color=black] (12) to  (6);
\draw [line width=0.625, color=black] (6) to  (7);
\draw [line width=0.625, color=black] (12) to  (13);
\draw [line width=0.625, color=black] (13) to  (7);
\draw [line width=0.625, color=black] (8) to  (7);
\draw [line width=0.625, color=black] (1) to  (8);
\draw [line width=0.625, color=black] (14) to  (15);
\draw [line width=0.625, color=black] (13) to  (14);
\draw [line width=0.625, color=black] (7) to  (15);
\draw [line width=0.625, color=black] (15) to  (16);
\draw [line width=0.625, color=black] (8) to  (16);
\end{tikzpicture}
}

\end{subfigure}%
\begin{subfigure}{.5\textwidth}
\centering
\usetikzlibrary{shapes.geometric}
\scalebox{0.8}{
\begin{tikzpicture}
[every node/.style={inner sep=0pt}]
\node (1) [circle,fill] at (175.0pt, -75.0pt) {\textcolor{black}{1}};
\node (2) [circle,fill] at (150.0pt, -100.0pt) {\textcolor{black}{2}};
\node (3) [circle,fill] at (200.0pt, -100.0pt) {\textcolor{black}{3}};
\node (4) [circle,fill] at (175.0pt, -125.0pt) {\textcolor{black}{4}};
\node (5) [circle,fill] at (200.0pt, -150.0pt) {\textcolor{black}{5}};
\node (6) [circle,fill] at (225.0pt, -125.0pt) {\textcolor{black}{6}};
\node (7) [circle,fill] at (150.0pt, -150.0pt) {\textcolor{black}{7}};
\node (8) [circle,fill] at (175.0pt, -175.0pt) {\textcolor{black}{8}};
\node (9) [circle,fill] at (162.5pt, -150.0pt) {\textcolor{black}{9}};
\node (10) [circle,fill] at (175.0pt, -150.0pt) {\textcolor{black}{1}};
\node (11) [circle,fill] at (187.5pt, -150.0pt) {\textcolor{black}{1}};
\node (12) [circle,fill] at (175.0pt, -100.0pt) {\textcolor{black}{1}};
\draw [line width=0.625, color=black] (7) to  (9);
\draw [line width=0.625, color=black] (9) to  (10);
\draw [line width=0.625, color=black] (10) to  (11);
\draw [line width=0.625, color=black] (11) to  (5);
\draw [line width=0.625, color=black] (5) to  (8);
\draw [line width=0.625, color=black] (7) to  (8);
\draw [line width=0.625, color=black] (5) to  (6);
\draw [line width=0.625, color=black] (4) to  (3);
\draw [line width=0.625, color=black] (4) to  (5);
\draw [line width=0.625, color=black] (6) to  (3);
\draw [line width=0.625, color=black] (4) to  (12);
\draw [line width=0.625, color=black] (3) to  (1);
\draw [line width=0.625, color=black] (2) to  (1);
\draw [line width=0.625, color=black] (2) to  (4);
\draw [line width=0.625, color=black] (12) to  (1);
\draw [line width=0.625, color=black] (4) to  (7);
\draw [line width=0.625, color=black] (4) to  (9);
\draw [line width=0.625, color=black] (4) to  (10);
\draw [line width=0.625, color=black] (4) to  (11);
\draw [line width=0.625, color=black] (4) to  (5);
\draw [line width=0.625, color=black] (9) to  (8);
\draw [line width=0.625, color=black] (10) to (8);
\draw [line width=0.625, color=black] (11) to (8);
\end{tikzpicture}
}
\end{subfigure}
\begin{subfigure}{.5\textwidth}
\centering
\usetikzlibrary{shapes.geometric}
\scalebox{0.8}{
\begin{tikzpicture}
[every node/.style={inner sep=0pt}]
\node (1) [circle, fill] at (87.5pt, -62.5pt) {\textcolor{black}{1}};
\node (2) [circle, fill] at (137.5pt, -62.5pt) {\textcolor{black}{2}};
\node (3) [circle, fill] at (112.5pt, -100.0pt) {\textcolor{black}{3}};
\node (4) [circle, fill] at (150.0pt, -125.0pt) {\textcolor{black}{4}};
\node (5) [circle, fill] at (137.5pt, -162.5pt) {\textcolor{black}{5}};
\node (6) [circle, fill] at (87.5pt, -162.5pt) {\textcolor{black}{6}};
\node (7) [circle, fill] at (75.0pt, -125.0pt) {\textcolor{black}{7}};
\node (8) [circle, fill] at (187.5pt, -100.0pt) {\textcolor{black}{8}};
\node (9) [circle, fill] at (225.0pt, -125.0pt) {\textcolor{black}{9}};
\node (10) [circle, fill] at (262.5pt, -75.0pt) {\textcolor{black}{0}};
\node (11) [circle, fill] at (287.5pt, -87.5pt) {\textcolor{black}{1}};
\node (12) [circle, fill] at (225.0pt, -175.0pt) {\textcolor{black}{2}};
\node (13) [circle, fill] at (250.0pt, -212.5pt) {\textcolor{black}{1}};
\node (14) [circle, fill] at (187.5pt, -200.0pt) {\textcolor{black}{4}};
\node (15) [circle, fill] at (212.5pt, -237.5pt) {\textcolor{black}{1}};
\node (17) [circle, fill] at (112.5pt, -200.0pt) {\textcolor{black}{7}};
\node (16) [circle, fill] at (75.0pt, -200.0pt) {\textcolor{black}{6}};
\node (18) [circle, fill] at (75.0pt, -237.5pt) {\textcolor{black}{8}};
\node (19) [circle, fill] at (112.5pt, -237.5pt) {\textcolor{black}{1}};
\draw [line width=0.625, color=black] (1) to  (3);
\draw [line width=0.625, color=black] (2) to  (3);
\draw [line width=0.625, color=black] (1) to  (2);
\draw [line width=0.625, color=black] (7) to  (3);
\draw [line width=0.625, color=black] (3) to  (4);
\draw [line width=0.625, color=black] (6) to  (7);
\draw [line width=0.625, color=black] (6) to  (5);
\draw [line width=0.625, color=black] (5) to  (4);
\draw [line width=0.625, color=black] (3) to  (6);
\draw [line width=0.625, color=black] (5) to  (3);
\draw [line width=0.625, color=black] (7) to  (4);
\draw [line width=0.625, color=black] (5) to  (7);
\draw [line width=0.625, color=black] (6) to  (4);
\draw [line width=0.625, color=black] (5) to  (17);
\draw [line width=0.625, color=black] (17) to  (19);
\draw [line width=0.625, color=black] (16) to  (17);
\draw [line width=0.625, color=black] (18) to  (16);
\draw [line width=0.625, color=black] (18) to  (19);
\draw [line width=0.625, color=black] (19) to  (16);
\draw [line width=0.625, color=black] (18) to  (17);
\draw [line width=0.625, color=black] (14) to  (12);
\draw [line width=0.625, color=black] (12) to  [in=119, out=307] (13);
\draw [line width=0.625, color=black] (12) to  (15);
\draw [line width=0.625, color=black] (12) to  (9);
\draw [line width=0.625, color=black] (13) to  (15);
\draw [line width=0.625, color=black] (15) to  (14);
\draw [line width=0.625, color=black] (8) to  [in=143, out=331] (9);
\draw [line width=0.625, color=black] (8) to  (12);
\draw [line width=0.625, color=black] (8) to  (2);
\draw [line width=0.625, color=black] (9) to  (10);
\draw [line width=0.625, color=black] (9) to  (11);
\draw [line width=0.625, color=black] (14) to  (13);
\end{tikzpicture}
}
\end{subfigure}
\end{figure}
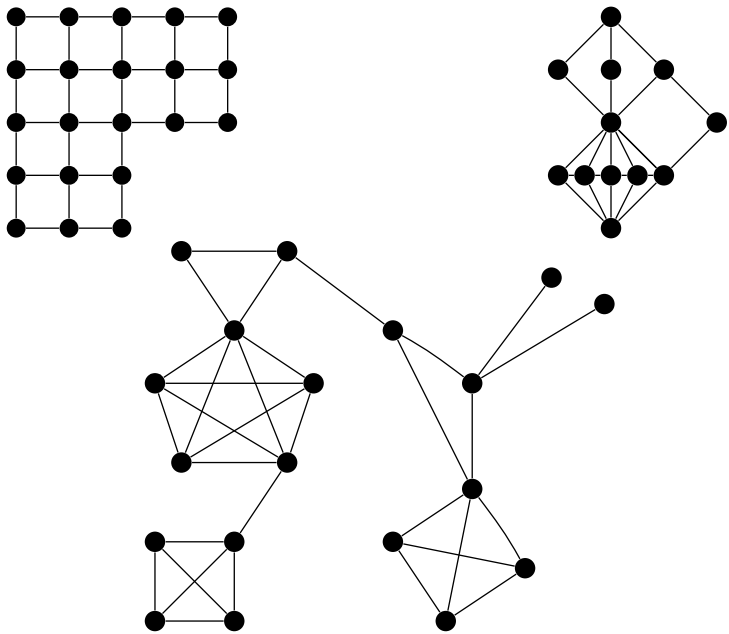

It is known that chain pursuit in the continuous Euclidean plane converges to the shortest path (see \cite{ants1}). One of the primary reasons for this is that the plane has no holes--it is simply connected. In a simply connected subspace of the plane, any path can be ``optimized'' into a shortest path by making small local deformations to it, as these deformations are never prevented by the undue presence of holes. When working with graphs, we can capture the notion of making small local changes to pursuit walks via the $\Delta$-deformations (and this led to propositions \ref{convergentdeform} and \ref{stabledeform}), but capturing the notion of ``holes'' in the right way is trickier, leading us to consider properties that are more indirect. One idea is to consider properties of convex shapes, as any convex shape is in particular holeless. 

Helly's theorem is a theorem about intersections of convex shapes, stated as follows: let $X_1 \ldots X_n$ be a collection of $n$ convex, finite subsets of $R^d$. If the intersection of every $d+1$ of these sets is nonempty, then the collection has a nonempty intersection (see \cite{helly} and Chapter 1 of \cite{helly2} for additional background). Helly's theorem motivates one of the possible, equivalent definitions of pseudo-modular graphs:

\theoremstyle{definition}
\begin{definition}
\label{pseudomodulardef}
A graph $G$ is called \textit{pseudo-modular}, or ``3-Helly'', if any three pairwise intersecting disks of $G$ have a nonempty intersection. (A disk of radius $r$ about the vertex $v$ is the set of all vertices of distance $\leq r$ from $v$)
\end{definition}

Pseudo-modular graphs were introduced in \cite{pseudomodular1} as a generalization of several important classes of graphs in metric graph theory, such as the so-called ``median'', ``modular'' and ``distance-hereditary'' graphs (see \cite{metricgraphtheory} for a general survey). It is not hard to confirm that the grid graph is pseudo-modular, and so are many of its sub-graphs and many grid-like graphs. 

To begin we wish to prove the following:

\begin{proposition} 
\label{pseudomodularconvergent}
Pseudo-modular graphs are convergent.
\end{proposition}

\begin{proof}
By \ref{convergentT2} it suffices to show that $G$ is convergent for delay time $\Delta =2$ (indeed, for all of our proofs from this point onwards, we will implicitly assume that $\Delta =2$ unless stated otherwise). We show that any walk can be $2$-deformed to a shortest path via $2$-deformations.

The proof is by induction on the number of vertices in the walk, denoted by $n$. For the induction base, it is simple to see that any walk of length $3$ or less (i.e. with 3 vertices or less) from $s$ to $t$ can be $2$-deformed to a shortest path (it is either the shortest path, or a direct link exists from $s$ to $t$, and then $2$-deformation clearly leads to it!). Thus the statement holds for $n \leq 3$. 

Now assume that all walks of length $n-1$ can be $2$-deformed to an optimal path. Let $P = v_1 \ldots v_n$ be a walk with $n$ vertices. We will show $P$ can be 2-deformed to a shortest path.

First, as the sub-walk $v_2 \ldots v_n$ is of length $n-1$, we may 2-deform it to a shortest path. Assume wlog that it already is. Then either $P$ is a shortest path (and we are done), or $v_1$ and $v_q$, for some $q > 1$, are discrepancy vertices. If $q \neq n$ then we can 2-deform the sub-walk from $v_1$ to $v_q$ into a shortest path (due to the inductive assumption), which shortens $P$, and reduces us to a previous case of the induction. So we simply need to handle the case where $v_1$ and $v_n$ are discrepancy vertices. 

If $v_1 = v_n$ (i.e. $P$ is a ``loop'' around $v_1$) then it follows that $|P| = 3$, so we are reduced to an earlier case of the induction. Otherwise, write $d = d(v_1, v_n)$. Consider the three disks $\mathcal{D}(v_1, 1)$, $\mathcal{D}(v_n, d-1)$, $\mathcal{D}(v_3, 1)$, where $\mathcal{D}(v,r)$ is the disk of radius $r$ about vertex $v$. 

We show that the three disks intersect pairwise:

\begin{enumerate}
\item $\mathcal{D}(v_1, 1)$ and $\mathcal{D}(v_{3}, 1)$ intersect at $v_{2}$. 

\item $\mathcal{D}(v_1, 1)$ and $\mathcal{D}(v_n, d-1)$ intersect at a vertex $u$ that lies along the shortest path from $v_1$ to $v_n$.

\item By Lemma \ref{discrepobserve}, 2 we have that $n-1 \leq d+2$, and therefore $d(v_4, v_n) \leq n-4 \leq d-1$. Hence $\mathcal{D}(v_{3}, 1)$ and $\mathcal{D}(v_n, d-1)$ intersect at $v_4$.
\end{enumerate}

Since $G$ is pseudo-modular, we learn from this that the three disks have a non-empty intersection. Thus, there exists a vertex $x$ for which (i) $d(v_1,x) \leq 1$, (ii) $d(x,v_3) \leq 1$, and (iii) $d(x,v_n) \leq d-1$,. It follows from (i) and (ii) that can 2-deform $v_1v_2v_3$ to $v_1xv_3$. Then, by the inductive assumption, we can 2-deform the sub-walk $x\ldots v_n$ to a shortest path from $x$ to $v_n$. Since by (iii), $x$ already lies on a shortest path from $v_1$ to $v_n$, this turns $P$ into a shortest path from $v_1$ to $v_n$ as desired.
\end{proof}

\setcounter{figure}{4} 
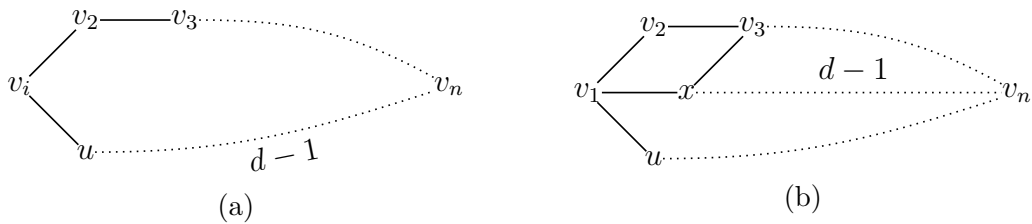
\begin{figure}
\caption{The constructions in the proof of Proposition \ref{pseudomodularconvergent}}
\centering
\begin{subfigure}{.5\textwidth}
\centering

\begin{tikzpicture}
[every node/.style={inner sep=0pt}]
\node (1) at (100.0pt, -137.5pt) {\textcolor{black}{$v_i$}};
\node (2) at (125.0pt, -112.5pt) {\textcolor{black}{$v_{2}$}};
\node (3) at (162.5pt, -112.5pt) {\textcolor{black}{$v_{3}$}};
\node (4) at (262.5pt, -137.5pt) {\textcolor{black}{$v_n$}};
\node (5) at (125.0pt, -162.5pt) {\textcolor{black}{$u$}};
\draw [line width=0.625, color=black] (1) to  (2);
\draw [line width=0.625, color=black] (1) to  (5);
\draw [line width=0.625, dotted, color=black] (5) to  [in=198, out=0] (4);
\draw [line width=0.625, color=black] (2) to  (3);
\draw [line width=0.625, dotted, color=black] (3) to  [in=151, out=0] (4);
\node at (200.0pt, -163.75pt) [rotate=11] {\textcolor{black}{$d-1$}};
\node at (211.875pt, -108.75pt) [rotate=348] {\textcolor{black}{}};
\end{tikzpicture}
\subcaption{}
\end{subfigure}%
\begin{subfigure}{.6\textwidth}
\centering
\usetikzlibrary{shapes.geometric}

\begin{tikzpicture}
[every node/.style={inner sep=0pt}]
\node (1) at (100.0pt, -137.5pt) {\textcolor{black}{$v_1$}};
\node (2) at (125.0pt, -112.5pt) {\textcolor{black}{$v_2$}};
\node (3) at (162.5pt, -112.5pt) {\textcolor{black}{$v_3$}};
\node (4) at (262.5pt, -137.5pt) {\textcolor{black}{$v_n$}};
\node (5) at (125.0pt, -162.5pt) {\textcolor{black}{$u$}};
\node (6) at (137.5pt, -137.5pt) {\textcolor{black}{$x$}};
\draw [line width=0.625, color=black] (1) to  (5);
\draw [line width=0.625, dotted, color=black] (5) to  [in=198, out=0] (4);
\draw [line width=0.625, color=black] (2) to  (3);
\draw [line width=0.625, dotted, color=black] (3) to  [in=151, out=0] (4);
\draw [line width=0.625, color=black] (1) to  (2);
\draw [line width=0.625, color=black] (1) to  (6);
\draw [line width=0.625, color=black] (6) to  (3);
\draw [line width=0.625, dotted, color=black] (6) to  (4);
\node at (200.625pt, -128.75pt) {\textcolor{black}{$d-1$}};
\end{tikzpicture}
\subcaption{}
\end{subfigure}
\end{figure}

We prove next that every pseudomodular graph is stable.

\begin{proposition} 
\label{pseudomodularstable}
Pseudo-modular graphs are stable.
\end{proposition}

For shorthand, a ``closed communicating class between $s$ and $t$'' refers to a closed communicating class of the the Markov chain $\mathcal{M}$ defined over the walks from $s$ to $t$ in $G$ ($\Delta =2$). 

\begin{proof}

We want to show that there is a unique closed communicating class for every $s$, $t$.

Assume for contradiction that $G$ is not stable. Then there are two vertices $s$, $t$ such that there are at least two distinct closed communicating classes from $s$ to $t$ (i.e. of the Markov chain $\mathcal{M}_2(s,t)$). Let $\mathcal{C}_1$ and $\mathcal{C}_2$ be two such classes. Note that $\mathcal{C}_1$ and $\mathcal{C}_2$ contain only \textit{shortest paths} from $s$ to $t$, as shown in Proposition \ref{pseudomodularconvergent}. 

The proof proceeds by induction on the distance between $s$ and $t$, $d(s,t)$. For the base case, if $d(s,t) \leq 2$ then there clearly must be just one closed communicating class between $s$ and $t$, a contradiction to $C_1$ and $C_2$ being distinct.

To proceed, assume that the statement holds for distances $\leq n-1$, and consider two vertices $s$ and $t$ such that $d(s,t) = n$.

Let $P_1 = v_1 v_2 \ldots v_{n+1} \in \mathcal{C}_1$ and $P_2 = u_1 u_2 \ldots u_{n+1} \in \mathcal{C}_2$, with $s = v_1 = u_1$ and $t = v_{n+1} = u_{n+1}$, be two shortest paths from $s$ to $t$. Consider the disks $\mathcal{D}(u_{2}, 1)$, $\mathcal{D}(v_{2}, 1)$ and $\mathcal{D}(t, n-2)$. We have that $\mathcal{D}(t, n-2)$ intersects $\mathcal{D}(u_{2}, 1)$ and $\mathcal{D}(v_{2}, 1)$ respectively at $u_3$ and $v_3$. Furthermore $s \in \mathcal{D}(u_{2}, 1) \cap \mathcal{D}(v_{2}, 1)$. Thus there must be a vertex $x$ in the intersection of all disks. For this vertex we have: $d(x,v_2) = 1$, $d(x,u_2) = 1$ and $d(x,t) = n-2$. Since $d(v_2,t) = d(u_2,t) = n-1$, we have that $x$ lies on a shortest path from both $u_2$ and $v_2$, to $t$. Let $P_x = x \ldots t$ be some fixed shortest path from $x$ to $t$. By the inductive assumption we can 2-deform the sub-walks $v_2 \ldots v_n$ of $P_1$ and $u_2 \ldots u_n$ of $P_2$ to the paths $v_2 P_x$ and $u_2 P_x$ respectively. Then we may deform the sub-walk $v_1 v_2 x$ (of the path $v_1v_2P_x$) to $v_1 u_2 x$, thus deforming the two paths into the same path. This contradicts that $\mathcal{C}_1$ and $\mathcal{C}_2$ are distinct, so we are done.
\end{proof}

An example of a graph which is stable but not pseudo-modular is seen in Figure 6.

\setcounter{figure}{5}
\begin{figure}[!htb]
\caption{Stable, non-pseudo-modular graph}
\centering
\usetikzlibrary{shapes.geometric}

\begin{tikzpicture}
[every node/.style={inner sep=0pt}]
\node (1) [circle, fill] at (112.5pt, -175.0pt) {\textcolor{black}{1}};
\node (2) [circle, fill] at (137.5pt, -150.0pt) {\textcolor{black}{2}};
\node (3) [circle, fill] at (162.5pt, -125.0pt) {\textcolor{black}{3}};
\node (4) [circle, fill] at (187.5pt, -150.0pt) {\textcolor{black}{4}};
\node (5) [circle, fill] at (212.5pt, -175.0pt) {\textcolor{black}{5}};
\node (6) [circle, fill] at (162.5pt, -175.0pt) {\textcolor{black}{6}};
\draw [line width=0.625, color=black] (1) to  (6);
\draw [line width=0.625, color=black] (1) to  (2);
\draw [line width=0.625, color=black] (6) to  (2);
\draw [line width=0.625, color=black] (4) to  (6);
\draw [line width=0.625, color=black] (4) to  (5);
\draw [line width=0.625, color=black] (6) to  (5);
\draw [line width=0.625, color=black] (2) to  (4);
\draw [line width=0.625, color=black] (2) to  (3);
\draw [line width=0.625, color=black] (4) to  (3);
\end{tikzpicture}
\end{figure}
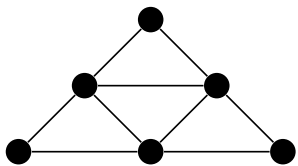

\newpage 

\subsection{Graph Products}

Graph products are a rich topic of study as well as an effective method of creating new topologies from old (see \cite{graphproducts, graphproducts2} for an overview). In this section we concern ourselves with two kinds of graph product operations, and show that they preserve graph convergence, and graph stability. 

\setcounter{figure}{6}
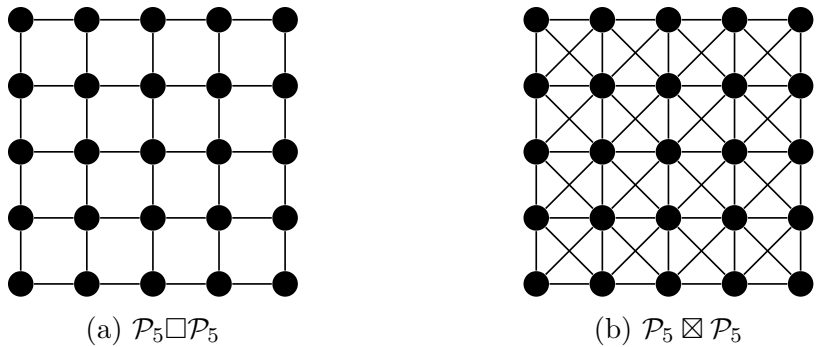
\begin{figure}[!htb]
\caption{Graph products of $\mathcal{P}_5$ with itself}
\begin{subfigure}{.5\textwidth}
\centering

\begin{tikzpicture}
[every node/.style={inner sep=0pt}]
\node (1) [circle, fill] at (125.0pt, -137.5pt) {\textcolor{black}{1}};
\node (2) [circle, fill] at (100.0pt, -137.5pt) {\textcolor{black}{2}};
\node (3) [circle, fill] at (100.0pt, -162.5pt) {\textcolor{black}{3}};
\node (4) [circle, fill] at (125.0pt, -162.5pt) {\textcolor{black}{4}};
\node (5) [circle, fill] at (100.0pt, -112.5pt) {\textcolor{black}{5}};
\node (6) [circle, fill] at (125.0pt, -112.5pt) {\textcolor{black}{6}};
\node (7) [circle, fill] at (150.0pt, -112.5pt) {\textcolor{black}{7}};
\node (8) [circle, fill] at (150.0pt, -137.5pt) {\textcolor{black}{8}};
\node (9) [circle, fill] at (100.0pt, -187.5pt) {\textcolor{black}{9}};
\node (10) [circle, fill] at (125.0pt, -187.5pt) {\textcolor{black}{0}};
\node (11) [circle, fill] at (100.0pt, -87.5pt) {\textcolor{black}{1}};
\node (12) [circle, fill] at (125.0pt, -87.5pt) {\textcolor{black}{2}};
\node (13) [circle, fill] at (150.0pt, -87.5pt) {\textcolor{black}{3}};
\node (14) [circle, fill] at (175.0pt, -87.5pt) {\textcolor{black}{4}};
\node (15) [circle, fill] at (175.0pt, -112.5pt) {\textcolor{black}{5}};
\node (16) [circle, fill] at (175.0pt, -137.5pt) {\textcolor{black}{6}};
\node (17) [circle, fill] at (75.0pt, -187.5pt) {\textcolor{black}{7}};
\node (18) [circle, fill] at (75.0pt, -162.5pt) {\textcolor{black}{8}};
\node (19) [circle, fill] at (75.0pt, -137.5pt) {\textcolor{black}{9}};
\node (20) [circle, fill] at (75.0pt, -112.5pt) {\textcolor{black}{0}};
\node (21) [circle, fill] at (75.0pt, -87.5pt) {\textcolor{black}{1}};
\node (22) [circle, fill] at (150.0pt, -162.5pt) {\textcolor{black}{2}};
\node (23) [circle, fill] at (175.0pt, -162.5pt) {\textcolor{black}{3}};
\node (24) [circle, fill] at (150.0pt, -187.5pt) {\textcolor{black}{4}};
\node (25) [circle, fill] at (175.0pt, -187.5pt) {\textcolor{black}{5}};
\draw [line width=0.625, color=black] (17) to  (18);
\draw [line width=0.625, color=black] (18) to  (19);
\draw [line width=0.625, color=black] (19) to  (20);
\draw [line width=0.625, color=black] (20) to  (21);
\draw [line width=0.625, color=black] (21) to  (11);
\draw [line width=0.625, color=black] (11) to  (5);
\draw [line width=0.625, color=black] (5) to  (2);
\draw [line width=0.625, color=black] (2) to  (3);
\draw [line width=0.625, color=black] (3) to  (9);
\draw [line width=0.625, color=black] (20) to  (5);
\draw [line width=0.625, color=black] (19) to  (2);
\draw [line width=0.625, color=black] (18) to  (3);
\draw [line width=0.625, color=black] (17) to  (9);
\draw [line width=0.625, color=black] (3) to  (4);
\draw [line width=0.625, color=black] (9) to  (10);
\draw [line width=0.625, color=black] (2) to  (1);
\draw [line width=0.625, color=black] (1) to  (4);
\draw [line width=0.625, color=black] (10) to  (4);
\draw [line width=0.625, color=black] (1) to  (6);
\draw [line width=0.625, color=black] (5) to  (6);
\draw [line width=0.625, color=black] (11) to  (12);
\draw [line width=0.625, color=black] (12) to  (6);
\draw [line width=0.625, color=black] (6) to  (7);
\draw [line width=0.625, color=black] (12) to  (13);
\draw [line width=0.625, color=black] (13) to  (7);
\draw [line width=0.625, color=black] (8) to  (7);
\draw [line width=0.625, color=black] (1) to  (8);
\draw [line width=0.625, color=black] (14) to  (15);
\draw [line width=0.625, color=black] (13) to  (14);
\draw [line width=0.625, color=black] (7) to  (15);
\draw [line width=0.625, color=black] (15) to  (16);
\draw [line width=0.625, color=black] (8) to  (16);
\draw [line width=0.625, color=black] (4) to  (22);
\draw [line width=0.625, color=black] (22) to  (8);
\draw [line width=0.625, color=black] (23) to  (16);
\draw [line width=0.625, color=black] (22) to  (23);
\draw [line width=0.625, color=black] (24) to  (22);
\draw [line width=0.625, color=black] (25) to  (23);
\draw [line width=0.625, color=black] (24) to  (25);
\draw [line width=0.625, color=black] (10) to  (24);
\end{tikzpicture}
\subcaption{$\mathcal{P}_5 \square \mathcal{P}_5$}
\end{subfigure}%
\begin{subfigure}{.5\textwidth}
\centering

\begin{tikzpicture}
[every node/.style={inner sep=0pt}]
\node (1) [circle, fill] at (125.0pt, -137.5pt) {\textcolor{black}{1}};
\node (2) [circle, fill] at (100.0pt, -137.5pt) {\textcolor{black}{2}};
\node (3) [circle, fill] at (100.0pt, -162.5pt) {\textcolor{black}{3}};
\node (4) [circle, fill] at (125.0pt, -162.5pt) {\textcolor{black}{4}};
\node (5) [circle, fill] at (100.0pt, -112.5pt) {\textcolor{black}{5}};
\node (6) [circle, fill] at (125.0pt, -112.5pt) {\textcolor{black}{6}};
\node (7) [circle, fill] at (150.0pt, -112.5pt) {\textcolor{black}{7}};
\node (8) [circle, fill] at (150.0pt, -137.5pt) {\textcolor{black}{8}};
\node (9) [circle, fill] at (100.0pt, -187.5pt) {\textcolor{black}{9}};
\node (10) [circle, fill] at (125.0pt, -187.5pt) {\textcolor{black}{0}};
\node (11) [circle, fill] at (100.0pt, -87.5pt) {\textcolor{black}{1}};
\node (12) [circle, fill] at (125.0pt, -87.5pt) {\textcolor{black}{2}};
\node (13) [circle, fill] at (150.0pt, -87.5pt) {\textcolor{black}{3}};
\node (14) [circle, fill] at (175.0pt, -87.5pt) {\textcolor{black}{4}};
\node (15) [circle, fill] at (175.0pt, -112.5pt) {\textcolor{black}{5}};
\node (16) [circle, fill] at (175.0pt, -137.5pt) {\textcolor{black}{6}};
\node (17) [circle, fill] at (75.0pt, -187.5pt) {\textcolor{black}{7}};
\node (18) [circle, fill] at (75.0pt, -162.5pt) {\textcolor{black}{8}};
\node (19) [circle, fill] at (75.0pt, -137.5pt) {\textcolor{black}{9}};
\node (20) [circle, fill] at (75.0pt, -112.5pt) {\textcolor{black}{0}};
\node (21) [circle, fill] at (75.0pt, -87.5pt) {\textcolor{black}{1}};
\node (22) [circle, fill] at (150.0pt, -162.5pt) {\textcolor{black}{2}};
\node (23) [circle, fill] at (175.0pt, -162.5pt) {\textcolor{black}{3}};
\node (24) [circle, fill] at (150.0pt, -187.5pt) {\textcolor{black}{4}};
\node (25) [circle, fill] at (175.0pt, -187.5pt) {\textcolor{black}{5}};
\draw [line width=0.625, color=black] (17) to  (18);
\draw [line width=0.625, color=black] (18) to  (19);
\draw [line width=0.625, color=black] (19) to  (20);
\draw [line width=0.625, color=black] (20) to  (21);
\draw [line width=0.625, color=black] (21) to  (11);
\draw [line width=0.625, color=black] (11) to  (5);
\draw [line width=0.625, color=black] (5) to  (2);
\draw [line width=0.625, color=black] (2) to  (3);
\draw [line width=0.625, color=black] (3) to  (9);
\draw [line width=0.625, color=black] (20) to  (5);
\draw [line width=0.625, color=black] (19) to  (2);
\draw [line width=0.625, color=black] (18) to  (3);
\draw [line width=0.625, color=black] (17) to  (9);
\draw [line width=0.625, color=black] (3) to  (4);
\draw [line width=0.625, color=black] (9) to  (10);
\draw [line width=0.625, color=black] (2) to  (1);
\draw [line width=0.625, color=black] (1) to  (4);
\draw [line width=0.625, color=black] (10) to  (4);
\draw [line width=0.625, color=black] (1) to  (6);
\draw [line width=0.625, color=black] (5) to  (6);
\draw [line width=0.625, color=black] (11) to  (12);
\draw [line width=0.625, color=black] (12) to  (6);
\draw [line width=0.625, color=black] (6) to  (7);
\draw [line width=0.625, color=black] (12) to  (13);
\draw [line width=0.625, color=black] (13) to  (7);
\draw [line width=0.625, color=black] (8) to  (7);
\draw [line width=0.625, color=black] (1) to  (8);
\draw [line width=0.625, color=black] (14) to  (15);
\draw [line width=0.625, color=black] (13) to  (14);
\draw [line width=0.625, color=black] (7) to  (15);
\draw [line width=0.625, color=black] (15) to  (16);
\draw [line width=0.625, color=black] (8) to  (16);
\draw [line width=0.625, color=black] (4) to  (22);
\draw [line width=0.625, color=black] (22) to  (8);
\draw [line width=0.625, color=black] (23) to  (16);
\draw [line width=0.625, color=black] (22) to  (23);
\draw [line width=0.625, color=black] (24) to  (22);
\draw [line width=0.625, color=black] (25) to  (23);
\draw [line width=0.625, color=black] (24) to  (25);
\draw [line width=0.625, color=black] (10) to  (24);
\draw [line width=0.625, color=black] (17) to  (3);
\draw [line width=0.625, color=black] (9) to  (18);
\draw [line width=0.625, color=black] (9) to  (4);
\draw [line width=0.625, color=black] (10) to  (3);
\draw [line width=0.625, color=black] (24) to  (23);
\draw [line width=0.625, color=black] (25) to  (22);
\draw [line width=0.625, color=black] (10) to  (22);
\draw [line width=0.625, color=black] (24) to  (4);
\draw [line width=0.625, color=black] (22) to  (16);
\draw [line width=0.625, color=black] (23) to  (8);
\draw [line width=0.625, color=black] (4) to  (8);
\draw [line width=0.625, color=black] (22) to  (1);
\draw [line width=0.625, color=black] (3) to  (1);
\draw [line width=0.625, color=black] (4) to  (2);
\draw [line width=0.625, color=black] (18) to  (2);
\draw [line width=0.625, color=black] (3) to  (19);
\draw [line width=0.625, color=black] (2) to  (6);
\draw [line width=0.625, color=black] (1) to  (5);
\draw [line width=0.625, color=black] (1) to  (7);
\draw [line width=0.625, color=black] (8) to  (6);
\draw [line width=0.625, color=black] (8) to  (15);
\draw [line width=0.625, color=black] (16) to  (7);
\draw [line width=0.625, color=black] (19) to  (5);
\draw [line width=0.625, color=black] (2) to  (20);
\draw [line width=0.625, color=black] (20) to  (11);
\draw [line width=0.625, color=black] (5) to  (21);
\draw [line width=0.625, color=black] (5) to  (12);
\draw [line width=0.625, color=black] (6) to  (11);
\draw [line width=0.625, color=black] (6) to  (13);
\draw [line width=0.625, color=black] (7) to  (12);
\draw [line width=0.625, color=black] (7) to  (14);
\draw [line width=0.625, color=black] (15) to  (13);
\end{tikzpicture}
\subcaption{$\mathcal{P}_5 \boxtimes \mathcal{P}_5$}
\end{subfigure}
\label{productexample}
\end{figure}

\theoremstyle{definition}
\begin{definition}[Cartesian product] The Cartesian product of two graphs $G = (V_1, E_1)$ and $H = (V_2, E_2)$, written $G \square H$, is defined to be the graph whose vertices are of the form $u = (x,y)$ where $x \in G$, $y \in H$, and where there is an edge between $(x_1, y_1) \to (x_2, y_2)$ iff $x_1 = x_2$ and $y_1y_2 \in E_2$, or $y_1 = y_2$ and $x_1x_2 \in E_1$. 
\end{definition}

\theoremstyle{definition}
\begin{definition}[Strong product] The strong product of two graphs $G = (V_1, E_1)$ and $H = (V_2, E_2)$, written $G \boxtimes H$, is defined to be the graph whose vertices are of the form $u = (x,y)$ where $x \in G$, $y \in H$, and where there is an edge between $(x_1, y_1) \to (x_2, y_2)$ iff $x_1 = x_2$ and $y_1y_2 \in E_2$, or $y_1 = y_2$ and $x_1x_2 \in E_1$, or $x_1x_2 \in E_1$ and $y_1y_2 \in E_2$.
\end{definition}

Any vertex $v$ of a graph product of $G_1$ and $G_2$ can be described as a pair $(x,y)$. The \textit{projection} of $v$ onto $G_1$ is defined to be $x$, and its projection onto $G_2$ is defined to be $y$. The projection of a walk $P$ onto $G_i$ is the walk over $G_i$ that consists of the projections of the vertices of $P$ onto $G_i$ in order. We will have the notation $d_i(v, u)$ refer to the distance between the projections of the vertices $v, u$ onto $G_i$. We note that the distance between two vertices $v$ and $u$ over $G_1 \square G_2$ is simply the ``taxicab metric'' distance \cite{taxicab}; $d_{G_1 \square G_2}(u,v) = d_1(v,u) + d_2(v,u)$. In comparison, it follows from the definition of a strong product that $d_{G_1 \boxtimes G_2}(v,u) = max \big( d_1(v,u), d_2(v,u) \big)$. 

Let $\mathcal{P}_n$ be the path graph on $n$ vertices. Then $\mathcal{P}_n \square \mathcal{P}_n$ is the regular $n \times n$ grid and $\mathcal{P}_n \boxtimes \mathcal{P}_n$ is the grid with diagonals, see Figure \ref{productexample}. Therefore this section offers another, different generalization of the known results regarding grid graphs.

\subsubsection*{Cartesian products}

\begin{proposition} 
\label{cartesianconvergent}
$G_1 \square G_2$ is convergent iff $G_1$ and $G_2$ are convergent.
\end{proposition}

\begin{proof}
In one direction, assume  $G_1 \square G_2$ is convergent and let, wlog, $P = v_1v_2 \ldots v_n$ be a walk in $G_1$. Let $u \in G_2$ be an arbitrary vertex. The walk $P'$ whose ith vertex is $v_i' = (v_i, u)$, can be projected onto $G_1$ and its projection is $P$. Since $P'$ is 2-deformable to a shortest path from $(v_1,u)$ to $(v_n,u)$, it follows from this that $P$ is 2-deformable to a shortest path from $v_1$ to $v_n$ (note that any valid 2-deformation of $P'$ over $G_1 \square G_2$ changes only the $G_1$ coordinate component of the vertices, since 2-deformations never lengthen a path and changing the fixed $G_2$ component $u$ will do so). 

In the other direction, assume $G_1$ and $G_2$ are convergent. Let $P = v_1 \ldots v_n$ be a walk in $G_1 \square G_2$. We will show it is 2-deformable to a shortest path from $v_1$ to $v_n$.

Call an edge in $G_1 \square G_2$ an \textit{x-change} if it affects the $G_1$ coordinate component and leaves $G_2$ fixed, else call it a \textit{y-change}. Now let $u_1u_2u_3$ be some walk in $G_1 \square G_2$. By the definition of a Cartesian product, if $u_1u_2$ is a y-change and $u_2u_3$ is an x-change, we can 2-deform $u_1u_2u_3$ into $u_1u_2'u_3$ such that $u_1u_2'$ is an x-change and $u_2'u_3$ is a y-change. Vice-versa, this is also true. We will call such 2-deformations \textit{swaps}. (See Figure \ref{swapdeformationfig} for an illustration). 

\setcounter{figure}{7}
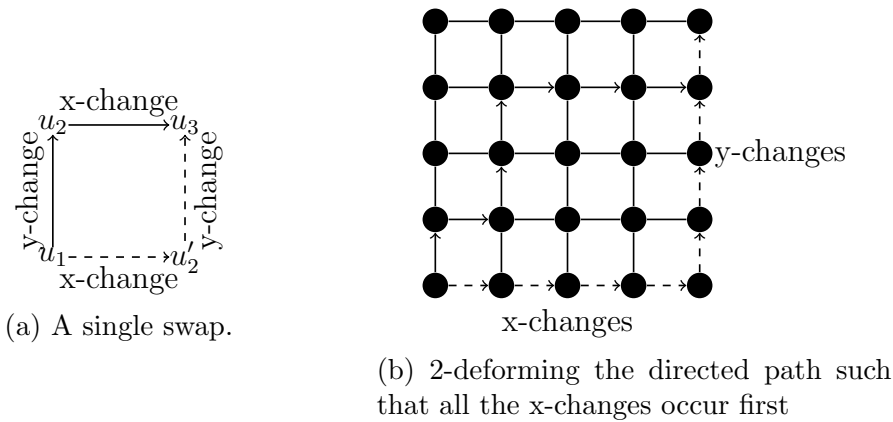
\begin{figure}[!htb]
\caption{The constructions in the proof of Proposition \ref{cartesianconvergent}.}
\centering
\begin{subfigure}{.5\textwidth}
\centering
\usetikzlibrary{shapes.geometric}

\begin{tikzpicture}
[every node/.style={inner sep=0pt}]
\node (1) at (150.0pt, -212.5pt) {\textcolor{black}{$u_1$}};
\node (2) at (200.0pt, -212.5pt) {\textcolor{black}{$u_2'$}};
\node (3) at (150.0pt, -162.5pt) {\textcolor{black}{$u_2$}};
\node (4) at (200.0pt, -162.5pt) {\textcolor{black}{$u_3$}};
\draw [line width=0.625, ->, color=black] (1) to  (3);
\draw [line width=0.625, ->, color=black] (3) to  (4);
\draw [line width=0.625, ->, dashed, color=black] (2) to  (4);
\draw [line width=0.625, ->, dashed, color=black] (1) to  (2);
\node at (141.25pt, -187.5pt) [rotate=90] {\textcolor{black}{y-change}};
\node at (175.0pt, -153.75pt) {\textcolor{black}{x-change}};
\node at (208.75pt, -187.5pt) [rotate=90] {\textcolor{black}{y-change}};
\node at (175.0pt, -221.25pt) {\textcolor{black}{x-change}};
\end{tikzpicture}
\subcaption{A single swap.}
\end{subfigure}%
\begin{subfigure}{.5\textwidth}
\centering

\begin{tikzpicture}
[every node/.style={inner sep=0pt}]
\node (1) [circle, fill] at (125.0pt, -137.5pt) {\textcolor{black}{1}};
\node (2) [circle, fill] at (100.0pt, -137.5pt) {\textcolor{black}{2}};
\node (3) [circle, fill] at (100.0pt, -162.5pt) {\textcolor{black}{3}};
\node (4) [circle, fill] at (125.0pt, -162.5pt) {\textcolor{black}{4}};
\node (5) [circle, fill] at (100.0pt, -112.5pt) {\textcolor{black}{5}};
\node (6) [circle, fill] at (125.0pt, -112.5pt) {\textcolor{black}{6}};
\node (7) [circle, fill] at (150.0pt, -112.5pt) {\textcolor{black}{7}};
\node (8) [circle, fill] at (150.0pt, -137.5pt) {\textcolor{black}{8}};
\node (9) [circle, fill] at (100.0pt, -187.5pt) {\textcolor{black}{9}};
\node (10) [circle, fill] at (125.0pt, -187.5pt) {\textcolor{black}{0}};
\node (11) [circle, fill] at (100.0pt, -87.5pt) {\textcolor{black}{1}};
\node (12) [circle, fill] at (125.0pt, -87.5pt) {\textcolor{black}{2}};
\node (13) [circle, fill] at (150.0pt, -87.5pt) {\textcolor{black}{3}};
\node (14) [circle, fill] at (175.0pt, -87.5pt) {\textcolor{black}{4}};
\node (15) [circle, fill] at (175.0pt, -112.5pt) {\textcolor{black}{5}};
\node (16) [circle, fill] at (175.0pt, -137.5pt) {\textcolor{black}{6}};
\node (17) [circle, fill] at (75.0pt, -187.5pt) {\textcolor{black}{7}};
\node (18) [circle, fill] at (75.0pt, -162.5pt) {\textcolor{black}{8}};
\node (19) [circle, fill] at (75.0pt, -137.5pt) {\textcolor{black}{9}};
\node (20) [circle, fill] at (75.0pt, -112.5pt) {\textcolor{black}{0}};
\node (21) [circle, fill] at (75.0pt, -87.5pt) {\textcolor{black}{1}};
\node (22) [circle, fill] at (150.0pt, -162.5pt) {\textcolor{black}{2}};
\node (23) [circle, fill] at (175.0pt, -162.5pt) {\textcolor{black}{3}};
\node (24) [circle, fill] at (150.0pt, -187.5pt) {\textcolor{black}{4}};
\node (25) [circle, fill] at (175.0pt, -187.5pt) {\textcolor{black}{5}};
\draw [line width=0.625, ->, color=black] (17) to  (18);
\draw [line width=0.625, color=black] (18) to  (19);
\draw [line width=0.625, color=black] (19) to  (20);
\draw [line width=0.625, color=black] (20) to  (21);
\draw [line width=0.625, color=black] (21) to  (11);
\draw [line width=0.625, color=black] (11) to  (5);
\draw [line width=0.625, color=black] (3) to  (9);
\draw [line width=0.625, color=black] (20) to  (5);
\draw [line width=0.625, color=black] (19) to  (2);
\draw [line width=0.625, ->, color=black] (18) to  (3);
\draw [line width=0.625, ->, dashed, color=black] (17) to  (9);
\draw [line width=0.625, color=black] (3) to  (4);
\draw [line width=0.625, ->, dashed, color=black] (9) to  (10);
\draw [line width=0.625, color=black] (2) to  (1);
\draw [line width=0.625, color=black] (1) to  (4);
\draw [line width=0.625, color=black] (10) to  (4);
\draw [line width=0.625, color=black] (1) to  (6);
\draw [line width=0.625, ->, color=black] (5) to  (6);
\draw [line width=0.625, color=black] (11) to  (12);
\draw [line width=0.625, color=black] (12) to  (6);
\draw [line width=0.625, ->, color=black] (6) to  (7);
\draw [line width=0.625, color=black] (12) to  (13);
\draw [line width=0.625, color=black] (13) to  (7);
\draw [line width=0.625, color=black] (8) to  (7);
\draw [line width=0.625, color=black] (1) to  (8);
\draw [line width=0.625, color=black] (13) to  (14);
\draw [line width=0.625, ->, color=black] (7) to  (15);
\draw [line width=0.625, color=black] (8) to  (16);
\draw [line width=0.625, color=black] (4) to  (22);
\draw [line width=0.625, color=black] (22) to  (8);
\draw [line width=0.625, ->, dashed, color=black] (23) to  (16);
\draw [line width=0.625, color=black] (22) to  (23);
\draw [line width=0.625, color=black] (24) to  (22);
\draw [line width=0.625, ->, dashed, color=black] (25) to  (23);
\draw [line width=0.625, ->, dashed, color=black] (24) to  (25);
\draw [line width=0.625, ->, dashed, color=black] (10) to  (24);
\draw [line width=0.625, ->, color=black] (3) to  (2);
\draw [line width=0.625, ->, color=black] (2) to  (5);
\draw [line width=0.625, ->, dashed, color=black] (15) to  (14);
\draw [line width=0.625, ->, dashed, color=black] (16) to  (15);
\node at (125.0pt, -201.875pt) {\textcolor{black}{x-changes}};
\node at (205.625pt, -137.5pt) {\textcolor{black}{y-changes}};
\end{tikzpicture}
\subcaption{2-deforming the directed path such that all the x-changes occur first}
\end{subfigure}
\label{swapdeformationfig}
\end{figure}

The idea is this: take the walk $P$ and perform swaps on its vertices until all x-changes are consecutive, and all y-changes are consecutive. This results in a walk $P'$ that can be divided into two sub-walks: $P_1' = u_1 \ldots u_k$ and $P_2' = u_k \ldots u_n$, such that the vertices in $P_1'$ have their $G_2$-component held constant, and the vertices in $P_2'$ have their $G_1$-component held constant. We can then use the convergence of $G_1$ and $G_2$ to 2-deform these two components to shortest paths, $P_1^*$ and $P_2^*$. It is simple to see from the definition of a Cartesian product that $P_1^*P_2^*$ must be a shortest path from $u_1$ to $u_n$; so we are done.
\end{proof}

\begin{proposition} 
\label{cartesianstable}
 $G_1 \square G_2$ is stable iff $G_1$ and $G_2$ are stable.
\end{proposition}

\begin{proof}
The direction where $G_1 \square G_2$ is stable is similar to its counterpart in Proposition \ref{cartesianconvergent}. We let $P = v_1 v_2 \ldots v_n$ be a walk in $G_1$ and create a walk $P'$ (as in \ref{cartesianconvergent}) whose projection onto $G_1$ is $P$. Since from stability it follows that $P'$ is 2-deformable to any shortest path from $(v_1,u)$ to $(v_n,u)$, it follows from this that $P$ is 2-deformable to any shortest path from $v_1$ to $v_n$. 

In the other direction, assume $G_1$ and $G_2$ are stable. Let $P = (x_1,y_1) \to (x_n,y_n)$ be any shortest path in $G_1 \square G_2$. We show that $P$ can always be deformed into a specific shortest path $Q$, thus showing that $G_1 \square G_2$ is stable (this shows stability via proposition \ref{stabledeform}). Let $v_i = (x_i, y_i)$ and let $Q_x$ be a fixed, arbitrary optimal path from $x_1 \to x_n$. Let $Q_y$ be a fixed, arbitrary optimal path from $y_1 \to y_n$. Then $Q$ is defined to be a shortest path of the form $(x_1,y_1)(x_2,y_1)\ldots(x_n,y_1)(x_n,y_2)\ldots(x_n,y_n)$ (that is, first we go through the path $Q_x$, holding the y-coordinate fixed, and then through $Q_y$, holding the x-coordinate fixed).

To deform $P$ to $Q$, we perform swaps on $P$ as in \ref{cartesianconvergent} to move all the x-changes to the beginning of the path, y-changes to the end, to get a path $P' = P_1'P_2'$. Then similarly to \ref{cartesianconvergent} we can use the stability of $G_1$ and $G_2$ to deform the front and end to paths $P_1'$,$P_2'$ to $Q_x$ and $Q_y$ respectively. Thus we deform $P$ to $Q$ as desired.
\end{proof}

\subsubsection*{Strong products}

For the rest of this section, by $d(v, u)$ we mean the usual distance from $v$ to $u$ over $G_1 \boxtimes G_2$. The fact that $d(v,u) = max \big( d_1(v,u), d_2(v,u) \big)$ is important and will be used extensively, sometimes implicitly, in the arguments below.

An important thing to note about walks in $G_1 \boxtimes G_2$ is that their $x$ and $y$ components can be 2-deformed independently, so long as the deformation doesn't shorten the walk. More explicitly, consider the walk $P = v_1 \ldots v_n$ and denote $v_i = (x_i, y_i)$. Note that if $x_ix_{i+1}x_{i+2}$ is 2-deformable to $x_ix'x_{i+2}$ then $v_iv_{i+1}v_{i+2} = (x_i,y_i)(x_{i+1},y_{i+1})(x_{i+2},y_{i+2})$ is 2-deformable to $(x_i,y_i)(x',y_{i+1}),(x_{i+2},y_{i+2})$. Thus we have changed the projection of $P$ onto $G_1$ without affecting the projection onto $G_2$. Equivalently, we can change the projection onto $G_2$ without changing the projection onto $G_1$. We will call $\Delta$-deformations that leave either the projection to $G_1$ or to $G_2$ unaffected \textit{independent} deformations.

We prove next the following property of strong products on graphs:

\begin{proposition} 
\label{strongproductconvergent}
$G_1 \boxtimes G_2$ is convergent iff $G_1$ and $G_2$ are convergent.
\end{proposition}

To aid us we will employ a useful definition. 

\begin{definition}
Let $P = v_1\ldots v_n$ be some walk of $G_1 \boxtimes G_2$. We define the \textit{x-score} of $v_i$, $1 < i < n$, to be $d_1(v_{i-1},v_{i+1}) - d_1(v_i,v_{i+1})$, and the \textit{y-score} to be $d_2(v_{i-1},v_{i+1}) - d_2(v_{i},v_{i+1})$.
\end{definition}

Note that the x-score (y-score) receive values only in -1, 0, or 1. The x-score (y-score) of a vertex $v_i$ of the walk $P = v_1\ldots v_n$  is a measure of how much ``closer'' $v_i$ brings us to $v_{i+1}$ when projected onto $G_i$, relative to $v_{i-1}$. If it is positive, then the projection of $v_i$ is closer to $v_{i+1}$ than was that of $v_{i-1}$. 

We start with an observation:

\begin{lemma}
\label{xyscoreobserve}
Let $\mathcal{C}$ be a closed communicating class of $G_1 \boxtimes G_2$ (i.e., of the Markov chain $\mathcal{M}_2(s,t)$ for some choice of $s$ and $t$). Then for any walk $P = v_1\ldots v_n \in \mathcal{C}$:
\begin{enumerate}
\item If $d_1(v_1,v_3) = 2$ and $d_2(v_1,v_3) \leq 1$, the x-score of every $v_i$ ($1 < i < n$) is 1
\item If $d_2(v_1,v_3) = 2$ and $d_1(v_1,v_3) \leq 1$, the y-score of every $v_i$ ($1 < i < n$) is 1
\item If $d_2(v_1,v_3) = 2$ and $d_1(v_1,v_3) = 2$, then either all $v_i$ have positive x-score, or all $v_i$ have positive y-score
\end{enumerate}
\end{lemma}

Since $P$ is $\Delta$-optimal, one of (1), (2) or (3) must hold. Essentially, Lemma \ref{xyscoreobserve} says that one of the projections of a path $P \in \mathcal{C}$ onto either $G_1$ or $G_2$ must be $\Delta$-optimal (in $G_1$ or $G_2$ respectively), since this is implied by either all of the x-scores or all of the y-scores being positive. The proof idea is to show that whenever this is not the case, $P$ can be ``stretched'' along either the $G_1$ or $G_2$ axes (via $\Delta$-deformation) to a path that is not $\Delta$-optimal, contradicting that $\mathcal{C}$ is a closed communicating class and hence contains only $\Delta$-optimal paths.

\begin{proof}
For proof of (1), suppose that $d_1(v_1,v_3) = 2$ and $d_2(v_1,v_3) \leq 1$. The proof is by contradiction: Let $v_k$ be the first vertex with non-positive x-score (we assume for contradiction that it exists), meaning that $d_1(v_{k-1}, v_{k+1}) - d_1(v_k, v_{k+1}) \leq 0$, or (alternatively written) $d_1(v_{k-1}, v_{k+1}) \leq d_1(v_k, v_{k+1})$. Recall the definition of 2-optimality, and that every walk in a closed communicating class is 2-optimal. Note that $P$ is 2-optimal, thus, for all $i$, $d(v_i,v_{i+2}) = 2$ and $d(v_i, v_{i+1}) = 1$. In particular we have that $d(v_k, v_{k+1}) = 1$, which implies that $d_1(v_{k-1}, v_{k+1}) \leq d_1(v_k, v_{k+1}) \leq 1$. Since $d(u,v)$ is the maximum of $d_1(u,v)$ and $d_2(u,v)$, and $d(v_{k-1}, v_{k+1}) = 2$, we must then have that $d_2(v_{k-1}, v_{k+1}) = 2$. 

Denote $v_i = (x_i, y_i)$. Using independent deformations, we can 2-deform the vertices of $P$ to have maximized y-scores, as follows: whenever $d(y_i, y_{i+2}) \geq 1$, we 2-deform $y_iy_{i+1}y_{i+2}$ to $y_iy'y_{i+2}$ such that $d(y', y_{i+2}) = d(y_i, y_{i+2}) - 1$ (this is always possible since when $y_i$ and $y_{i+2}$ are not the same vertex there is always a vertex connected to both of them that gets closer to $y_{i+2}$). We do this without affecting the projection of $P$ on $G_1$ (and so the x-scores remain unchanged), creating a walk $P' = v_1v_2'\ldots v_{n-1}'v_n$. Since $d_2(v_1, v_3) \leq 1$, we then have that $d_2(v_i',v_{i+2}') \leq 1$ for all $i$. 

Recalling that $v_k$ has non-positive x-score (meaning that so does $v'_k$), we have that $d_1(v_{k-1}',v_{k+1}') \leq 1$ but now also $d_2(v_{k-1}',v_{k+1}') \leq 1$, and so $d(v_{k-1}',v_{k+1}') \leq 1$, a contradiction to the 2-optimality of all walks in $\mathcal{C}$ and the closure of $\mathcal{C}$ under 2-deformations.

The proof of (2) is the same.

For proof of (3), again by contradiction, let $v_k$ be the first vertex that doesn't have both x- and y-scores positive. Suppose wlog that the y-score is non-positive, meaning we must have $d_2(v_{k-1},v_{k+1}) \leq d_2(v_k,v_{k+1}) \leq 1$. Due to 2-optimality we must then have $d_1(v_{k-1},v_{k+1}) = 2$, so we can apply the same argument as (1) to the subwalk $v_{k-1} \ldots v_n$ (that is, our new '$v_1$' and '$v_3$' are $v_{k-1}$ and $v_{k+1}$), to get that the vertices past $v_k$ have positive x-score as well.
\end{proof}

We can now move on to the proof of \ref{strongproductconvergent}.

\begin{proof}
The case where $G_1 \boxtimes G_2$ is convergent uses the same idea as propositions \ref{cartesianconvergent} and \ref{cartesianstable}. Let $P = v_1v_2 \ldots v_n$ be a walk in $G_1$ and create a walk $P'$ over $G_1 \boxtimes G_2$ (as in \ref{cartesianconvergent}) whose projection onto $G_1$ is $P$ and whose y-coordinate is fixed. We may deform $P'$ into a shortest path, since $G_1 \boxtimes G_2$ is convergent. Any atomic 2-deformation of $P'$ that changes its projection onto $G_1$ can be translated into a 2-deformation of $P$ by looking only at the changes to the x-coordinates. This implies that $P$ is deformable to a shortest path in $G_1$.

In the other direction, suppose $G_1$ and $G_2$ are convergent. Let $\mathcal{C}$ be a closed communicating class of $G_1 \boxtimes G_2$ and let $P = v_1\ldots v_n$ be a walk in $\mathcal{C}$. We will show $P$ is a shortest path from $v_1$ to $v_n$.

Denote $v_i = (x_i, y_i)$. Assume $P$ is not optimal. Then $n-1 > d(v_1, v_n) \geq max\big(d_1(v_1,v_n), d_2(v_1,v_n)\big)$, and thus neither of the walks $X = x_1 x_2 \ldots x_n$ and $Y = y_1 y_2 \ldots y_n$ is optimal. Since $G_1$ and $G_2$ are convergent, we can 2-deform these paths to optimality. Hence, for both $X$ and $Y$, there is a sequence of atomic 2-deformations that eventually results in a walk of length $n-1$ (i.e. a length one less than their current length). The penultimate element of this sequence will be a walk $u_1 u_2 \ldots u_n$ such that $d(u_{k-1}, u_{k+1}) \leq 1$ for some $k$, i.e., a walk that is not 2-optimal (since we cannot delete any vertices from a 2-optimal walk via a single atomic 2-deformation).

Let $X' = x_1' x_2' \ldots x_n'$ be a 2-deformation of $X$ such that for some $k$, $d(x'_{k-1},x'_{k+1}) \leq 1$.  There must be such a 2-deformation in light of the above. In X', 2-deform the sub-walk $x'_{k-1}x'_kx'_{k+1}$ to $x'_{k-1}x'_{k-1}x'_{k+1}$, calling the new walk $X^*$. Note that $X^*$ has a vertex with non-positive (zero) x-score.

Let $Y'$ be a 2-deformation of $Y$ such that for some $j$, $d(y'_{j-1}, y'_{j+1}) \leq 1$. We again 2-deform the sub-walk $y'_{j-1}y'_jy'_{j+1}$ to $y'_{j-1}y'_{j-1}y'_{j+1}$, resulting in a walk $Y^*$ that has a vertex with non-positive (zero) y-score.

We independently deform the x- and y- components of $P$ into $X^*$ and $Y^*$ respectively. Call the new walk $P^*$.

According to Lemma \ref{xyscoreobserve}, $P^*$ either has all x-scores positive, or all y-scores positive. But we know this to be false due to the way we constructed $X^*$ and $Y^*$, a contradiction to Lemma \ref{xyscoreobserve}.
\end{proof}

\begin{proposition} 
\label{strongproductstable}
$G_1 \boxtimes G_2$ is stable iff $G_1$ and $G_2$ are stable.
\end{proposition}

\begin{proof}
The direction where $G_1 \boxtimes G_2$ is stable uses precisely the same idea as in Propositions \ref{cartesianconvergent}, \ref{cartesianstable}, and \ref{strongproductconvergent},  and is omitted.

Suppose $G_1$ and $G_2$ are stable. We will show $G_1 \boxtimes G_2$ is stable. To this end we show that every shortest path from $s = (s_x, s_y)$ to $t = (t_x, t_y)$ is deformable to a fixed path $Q$ (see Proposition \ref{stabledeform}). We separate the proof into cases. 

(1) Assume $d_1(s,t) = d_2(s,t)$. Set $Q = u_1 \ldots u_n$ to be some arbitrary optimal path from $s$ to $t$, and set $Q_x$, $Q_y$ to be the projections over $G_1$ and $G_2$ respectively.

Let $P = v_1 \ldots v_n$ be a path from $s = v_1$ to $t = v_n$. We can assume it is an optimal path due to convergence. Denote by $X = x_1 \ldots x_n$, $Y = y_1 \ldots y_n$ the projections of $P$ onto $G_1$ and $G_2$. Since $d_1(v_1,v_n) = d_2(v_1,v_n)$, we have that $X$ and $Y$ are both optimal. Thanks to the stability of $G_1$ and $G_2$, we can independently 2-deform $X$ to $Q_x$ and $Y$ to $Q_y$. So any $P$ from $s$ to $t$ is deformable to $Q$ and we are done.

(2) Assume $d_1(s,t) \neq d_2(s,t)$. wlog we will assume that $d_1(s,t) > d_2(s,t)$. Write $n_2 = d_2(s,t) + 1$ for shorthand. Define $Q$ similar to before, its projections over $G_1$ and $G_2$ being as follows: first, $Q_x$ is some arbitrary optimal path from $s_x$ to $t_x$. Then, the first $n_2$ vertices of $Q_y$ are some optimal path from $s_y$ to $t_y$ and the rest are $t_y$ repeated ($Q_y = \stackrel{n_2}{s_y \ldots t_y} t_y \ldots t_y$).

As before let $P = v_1 \ldots v_n$ be any shortest path from $s$ to $t$, with the projections $X$ and $Y$ defined as in the previous case. We can deform $X$ to $Q_x$ like before. Deforming $Y$ to $Q_y$ is more delicate.

First we show that given a subwalk $Y(j,k) = y_j \ldots y_k$ such that $d_2(y_j,y_k) < k - j$ (that is, the subwalk is sub-optimal), it is possible to deform $Y(j,k)$ such that both the $k$th and the $(k-1)$th vertices will be equal to $y_k$. 

Since $Y(j,k)$ is sub-optimal and $G_2$ is stable we can 2-deform $Y(j,k)$ to a walk $Y(j,k)' = y_j y_{j+1}' \ldots y_{k-1}' y_k$ such that $d_2 (y_t', y_{t+2}') = 1$ for some $t$ (the same idea was used in Proposition \ref{strongproductconvergent}). 

We perform a sequence of 2-deformations on $Y(j,k)'$, explained as follows: first we look at the sub-walk of length 3  $y_t'y_{t+1}'y_{t+2}'$ and 2-deform it to $y_t'y_{t+2}'y_{t+2}'$ (this is possible by the above). We then look at the sub-walk $y_{t+2}'y_{t+2}'y_{t+3}'$  that starts 1 vertex after $y_t'$ (which was $y_{t+1}'y_{t+2}'y_{t+3}'$ before this deformation), and 2-deform it to $y_{t+2}'y_{t+3}'y_{t+3}'$. We continue moving ``rightward'' in this manner, at every step taking a sub-walk of the form $y_{t+b}'y_{t+b}'y_{t+b+1}'$ and 2-deforming it to $y_{t+b}'y_{t+b+1}'y_{t+b+1}'$:

\begin{align*}
	\mathrm{y_j \ldots y_t' y_{t+1}' y_{t+2}' \ldots y_{k-1}' y_k} \stackrel{2-def.}{\rightsquigarrow} \mathrm{y_j \ldots y_t' y_{t+2}' y_{t+2}' \ldots y_k}\\
	\mathrm{y_j  \ldots y_t' y_{t+2}' y_{t+2}' y_{t+3}' \ldots  y_k} \rightsquigarrow \mathrm{y_j  \ldots y_t' y_{t+2}' y_{t+3}' y_{t+3}' \ldots  y_k}\\
	&\vdotswithin{=} \notag \\
	\mathrm{y_j  \ldots y_{k-1}' y_{k-1}'  y_k'} \rightsquigarrow \mathrm{y_j  \ldots y_{k-1}' y_k y_k}
\end{align*}

This eventually gives the desired deformation: a walk that ends with $y_ky_k$.

Let $o \geq n_2$ be the earliest index of the vertex $t_y$ in the walk $Y$. We deform $Y(1,o)$ a finite number of times using the above idea, to duplicate $t_y$ to the index $n_2$, resulting in a walk $Y'$. The last ($n$th) vertex of $Y'$ is $t_y$, and so is the $n_2$th vertex, so the subwalk $Y'(n_2,n)$ is sub-optimal. Thus we may repeatedly apply the above 2-deformations to duplicate $t_y$ across the rest of the walk. This leaves us with a 2-deformation of $Y$, $Y''$, such that the first $n_2$ vertices are an optimal path from $s_y$ to $t_y$, and the rest of the vertices are $t_y$. Finally we 2-deform the first $n_2$ vertices (using the stability of $G_2$) to equal those of $Q_y$, and we are done. 
\end{proof}

An interesting corollary of the above propositions is the following:

\begin{corollary}
$G_1 \boxtimes G_2$ is stable (resp., convergent) iff $G_1 \square G_2$ is stable (resp., convergent)
\end{corollary}

In other words, for questions of stability and convergence of probabilistic pursuit, there is no difference between the topology induced by the ``taxicab'' distance $d\big((x_1,y_1),(x_2,y_2)\big) = d_1(x_1,x_2) + d_2(y_1,y_2)$ and the topology induced by the distance $d\big((x_1,x_2),(y_1,y_2)\big) = max\big( d_1(x_1,y_1), d_2(x_2,y_2) \big)$.

\subsection{Planar and Chordal Graphs}

It is surprisingly difficult to pinpoint the effect of the underlying graph being \textit{planar} on chain pursuit. Certainly, not every planar graph is convergent - a simple counterexample is a cycle of length 5 and above. But one might expect that planar graphs with high connectivity, or other good ``regularity'' properties (for example, low complexity of the planar graph's faces), would be convergent, if not stable. The counterexamples in Figure 9 highlight the difficulty of pinpointing such properties. Figure 9, (a) shows a maximal planar graph that is not convergent. Figure 9, (b) shows a matchstick graph (a planar graph that can be drawn on the plane with all edge lengths being 1) whose faces are all triangular or square. The dashed edges show a suboptimal recurrent path from $s$ to $t$; the double lines show the optimal path.

\setcounter{figure}{8}
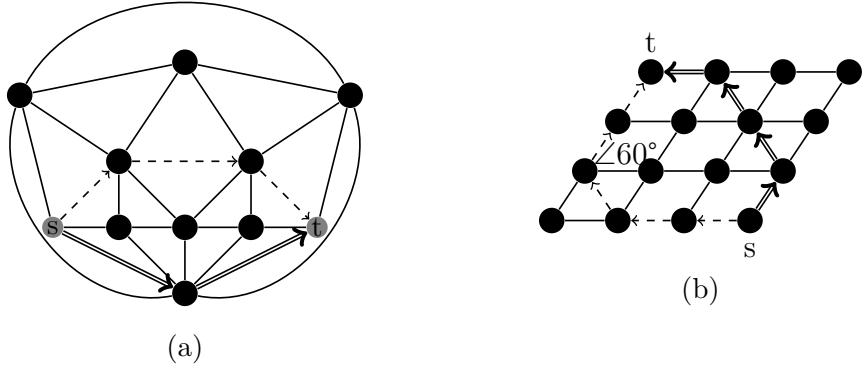
\begin{figure}[h]
\label{fig:planarexamples}
\caption{Some planar, non-convergent graphs.}
\begin{subfigure}{.5\textwidth}
\centering

\begin{tikzpicture}
[every node/.style={minimum size=8pt, inner sep=0pt}]
\node (2) [circle, fill] at (150.0pt, -137.5pt) {\textcolor{black}{1}};
\node (3) [circle, fill] at (200.0pt, -137.5pt) {\textcolor{black}{3}};
\node (5) [circle, fill] at (150.0pt, -162.5pt) {\textcolor{black}{5}};
\node (6) [circle, fill] at (175.0pt, -162.5pt) {\textcolor{black}{6}};
\node (7) [circle, fill] at (200.0pt, -162.5pt) {\textcolor{black}{7}};
\node (8) [circle, fill] at (175.0pt, -187.5pt) {\textcolor{black}{8}};
\node (9) [circle, fill] at (112.5pt, -112.5pt) {\textcolor{black}{9}};
\node (10) [circle, fill] at (175.0pt, -100.0pt) {\textcolor{black}{1}};
\node (1) [circle, fill=gray] at (125.0pt, -162.5pt) {\textcolor{black}{\small{s}}};
\node (4) [circle, fill=gray] at (225.0pt, -162.5pt) {\textcolor{black}{\small{t}}};
\node (11) [circle, fill] at (237.5pt, -112.5pt) {\textcolor{black}{1}};
\draw [line width=0.625, ->, dashed, color=black] (1) to  (2);
\draw [line width=0.625, ->, dashed, color=black] (2) to  (3);
\draw [line width=0.625, ->, dashed, color=black] (3) to  (4);
\draw [line width=0.625, ->, double, color=black] (1) to  (8);
\draw [line width=0.625, ->, double, color=black] (8) to  (4);
\draw [line width=0.625, color=black] (1) to  (5);
\draw [line width=0.625, color=black] (5) to  (6);
\draw [line width=0.625, color=black] (6) to  (7);
\draw [line width=0.625, color=black] (7) to  (4);
\draw [line width=0.625, color=black] (2) to  (6);
\draw [line width=0.625, color=black] (6) to  (3);
\draw [line width=0.625, color=black] (7) to  (3);
\draw [line width=0.625, color=black] (5) to  (2);
\draw [line width=0.625, color=black] (1) to  (9);
\draw [line width=0.625, color=black] (2) to  (9);
\draw [line width=0.625, color=black] (2) to  (10);
\draw [line width=0.625, color=black] (3) to  (10);
\draw [line width=0.625, color=black] (3) to  (11);
\draw [line width=0.625, color=black] (4) to  (11);
\draw [line width=0.625, color=black] (9) to  (10);
\draw [line width=0.625, color=black] (10) to  (11);
\draw [line width=0.625, color=black] (8) to  [in=288, out=349] (11);
\draw [line width=0.625, color=black] (8) to  [in=252, out=191] (9);
\draw [line width=0.625, color=black] (9) to  [in=119, out=61] (11);
\draw [line width=0.625, color=black] (5) to  (8);
\draw [line width=0.625, color=black] (6) to  (8);
\draw [line width=0.625, color=black] (7) to  (8);
\end{tikzpicture}
\subcaption{}
\end{subfigure}%
\begin{subfigure}{.5\textwidth}
\centering

\begin{tikzpicture}
[every node/.style={inner sep=0pt}]
\node (1) [circle, fill] at (125.0pt, -218.75pt) {\textcolor{black}{1}};
\node (2) [circle, fill] at (150.0pt, -218.75pt) {\textcolor{black}{2}};
\node (3) [circle, fill] at (137.5pt, -200.0pt) {\textcolor{black}{3}};
\node (4) [circle, fill] at (162.5pt, -200.0pt) {\textcolor{black}{4}};
\node (5) [circle, fill] at (112.5pt, -200.0pt) {\textcolor{black}{5}};
\node (6) [circle, fill] at (100.0pt, -218.75pt) {\textcolor{black}{6}};
\node (7) [circle, fill] at (87.5pt, -237.5pt) {\textcolor{black}{7}};
\node (8) [circle, fill] at (112.5pt, -237.5pt) {\textcolor{black}{8}};
\node (9) [circle, fill] at (137.5pt, -237.5pt) {\textcolor{black}{9}};
\node (10) [circle, fill] at (87.5pt, -200.0pt) {\textcolor{black}{1}};
\node (11) [circle, fill] at (75.0pt, -218.75pt) {\textcolor{black}{1}};
\node (12) [circle, fill] at (62.5pt, -237.5pt) {\textcolor{black}{2}};
\node (13) [circle, fill] at (100.0pt, -181.25pt) {\textcolor{black}{3}};
\node (14) [circle, fill] at (125.0pt, -181.25pt) {\textcolor{black}{4}};
\node (15) [circle, fill] at (150.0pt, -181.25pt) {\textcolor{black}{1}};
\node (16) [circle, fill] at (175.0pt, -181.25pt) {\textcolor{black}{1}};
\draw [line width=0.625, color=black] (1) to  (3);
\draw [line width=0.625, color=black] (2) to  (4);
\draw [line width=0.625, color=black] (1) to  (2);
\draw [line width=0.625, color=black] (3) to  (4);
\draw [line width=0.625, ->, double, color=black] (2) to  (3);
\draw [line width=0.625, ->, double, color=black] (9) to  (2);
\draw [line width=0.625, color=black] (8) to  (1);
\draw [line width=0.625, color=black] (7) to  (6);
\draw [line width=0.625, color=black] (6) to  (1);
\draw [line width=0.625, color=black] (6) to  (5);
\draw [line width=0.625, color=black] (5) to  (3);
\draw [line width=0.625, color=black] (7) to  (12);
\draw [line width=0.625, color=black] (12) to  (11);
\draw [line width=0.625, color=black] (6) to  (11);
\draw [line width=0.625, ->, dashed, color=black] (11) to  (10);
\draw [line width=0.625, color=black] (10) to  (5);
\draw [line width=0.625, ->, dashed, color=black] (10) to  (13);
\draw [line width=0.625, color=black] (5) to  (14);
\draw [line width=0.625, color=black] (3) to  (15);
\draw [line width=0.625, color=black] (14) to  (15);
\draw [line width=0.625, color=black] (15) to  (16);
\draw [line width=0.625, color=black] (4) to  (16);
\draw [line width=0.625, ->, double, color=black] (3) to  (14);
\draw [line width=0.625, ->, dashed, color=black] (7) to  (11);
\draw [line width=0.625, ->, dashed, color=black] (9) to  (8);
\draw [line width=0.625, ->, dashed, color=black] (8) to  (7);
\draw [line width=0.625, ->, double, color=black] (14) to  (13);
\node at (137.5pt, -248.75pt) {\textcolor{black}{s}};
\node at (100.0pt, -170.0pt) {\textcolor{black}{t}};
\node at (90.5pt, -211.875pt) {\textcolor{black}{$\angle\small60\degree$}};
\end{tikzpicture}
\subcaption{}
\end{subfigure}
\end{figure}

An outerplanar graph is a graph that has a planar drawing for which all vertices belong to the outer face of the drawing. A maximal outerplanar graph is an outerplanar graph to which we can add no edges. A positive result for planar graphs is the following:

\begin{proposition}
\label{maximalplanarconvergent}
Every maximal outerplanar graph is convergent.
\end{proposition}

In fact, this proposition stems from a more general observation. A fairly well-known class of graphs are the \textit{chordal graphs} (see \cite{chordalgraphs, chordalgraphs2} for background). Chordal graphs can be defined in two equivalent ways.

\theoremstyle{definition}
\begin{definition}
\label{simplicialvertexdefinition}
A \textit{simplicial vertex} is a vertex whose neighbors form a complete graph \cite{chordalgraphs}.
\end{definition}

\theoremstyle{definition}
\begin{definition}
\label{chordalgraphdefinition}
The following are equivalent characterizations of chordal graphs:
\begin{enumerate}
\item All cycles of four or more vertices of $G$ have a chord (an edge that is not part of the cycle but connects two vertices of the cycle).
\item $G$ has a \textit{perfect elimination ordering}: an ordering of the vertices of the graph such that, for each vertex v, v is a \textit{simplicial vertex} of the graph induced by $v$ and the vertices that occur after $v$ in the order.
\end{enumerate}
\end{definition}

Definition (1) hints at a ``regularity'' property of the kind we are looking for that is possessed by chordal graphs. 

It is well-known and simple to prove that every maximal outerplanar graph is chordal. To see this, note that the regions in the interior of a maximal outerplanar graph form a tree (if there was a cycle, it would necessarily surround some vertex of the graph, which contradicts outerplanarity). As such any region of the outerplanar graph corresponding to a leaf of that tree must have a vertex of degree 2. It is simple to see that this vertex is simplicial, and that after its removal the graph remains maximal outerplanar. Thus we found a perfect elimination ordering, and every such graph must be chordal. (See the famous ``ear clipping'' algorithm \cite{earclipping}). 

Thus, proposition \ref{maximalplanarconvergent} is a consequence of the following:

\begin{proposition}
\label{chordalconvergent}
Chordal graphs are convergent.
\end{proposition}

\begin{proof}
The proof is by induction. We assume every chordal graph of size $n-1$ is convergent, and prove that this yields the result for graphs of size $n$. In the base case, it is simple to verify that any graph (chordal or not) with $n \leq 4$  vertices is convergent.

Let $G$ be a chordal graph of size $n$ and let $v$ be a simplicial vertex of $G$. Consider a walk $P = u_1 \ldots u_m$ in $G$. We show that $P$ can be 2-deformed into a shortest path. We separate our proof into three cases:

(1) If $v$ does not occur as a vertex in $P$, then we can 2-deform $P$ into a shortest path just as we would working in $G - {v}$ (which is chordal and therefore convergent by our inductive assumption). (Note that since $v$ is simplicial, it does not occur in any shortest path from $u_1$ to $u_m$). 

(2) If $v$ occurs in $P$ as a vertex $u_i$, $1 < i < m$, then since it is simplicial we have $d(u_{i-1}, u_{i+1}) = 1$, thus we can 2-deform $u_{i-1}u_iu_{i+1}$ to $u_{i-1}u_{i+1}$ and remove $v$ from $P$. Thus we are reduced to either case (1) or case (3).

(3) Either $u_1 = v$ or $u_m = v$ and $v$ occurs nowhere else in $P$. Here we require another separation to cases: 

In the first case, both $u_1 = v$ and $u_m = v$. Using the convergence of $G-{v}$ we can 2-deform $u_2 \ldots u_{m-1}$ to $u_2 u_{m-1}$ (since $u_2,u_{m-1}$ are neighbors of $v$ and therefore $d(u_2,u_{m-1}) \leq 1$). This leaves us with the walk $P'=u_1u_2u_{m-1}u_m$, and we have $d(u_1,u_{m-1}) = 1$, so we may remove $u_2$ from it via 2-deformation. Finally since $d(u_1,u_m)=0$ we can remove $u_{m-1}$, leaving us with $P''=u_1u_m$, an optimal path.

In the other case, we assume wlog that only $u_1 = v$ (the case where $u_m = v$ is symmetrical). By the inductive assumption, and since $G-{v}$ is chordal, we can 2-deform the subwalk $u_2 u_3 \ldots u_m$ to a shortest path. We will assume wlog that it already is.  Since $\Delta = 2$, we can assume that $|P| = m > 3$, otherwise the proof is trivial. 

Suppose that in spite of $u_2 u_3 \to \ldots u_m$ being optimal, the walk $P$ is not a shortest path from $u_1$ to $u_m$. Recall the definition of discrepancy vertices. Since it is not optimal, $P$ must contain a pair of discrepancy vertices. Since the subwalk $u_2 \ldots u_m$ is optimal, this must be a pair of the form $u_1, u_k$ for some $k$ (as optimal subwalks contain no discrepancy vertices).

For simplicity, will assume that $k = m$, and deal with the case where $k \neq m$ at the end. That is, we assume that $u_1$ and $u_m$ are the discrepancy vertices. 

Let $P^* = v_1 \ldots v_{l}$ be a shortest path from $v_1 = u_1 = v$ to $v_{l} = u_m$. Since $v$ is a simplicial vertex and $u_2$ is a neighbor of $v$, we have that $d(u_2, u_m) \leq d(v, u_m)$. Note that since the sub-walk $u_2 u_3 \ldots u_m$ is optimal, and goes through $m-2$ edges, we have that $m-2 = d(u_2, u_m)$. In turn this implies that $|P^*| = d(v, u_m) + 1 \geq d(u_2, u_m) + 1 = m-1 \geq 3$.

Since $u_1$ and $u_m$ are discrepancy vertices, the paths $P^*$ and $P$ contain no shared vertices except at the endpoints (see Lemma \ref{discrepobserve}). Thus the vertices of $P \cup P^*$ form a cycle. Since $|P^*|, |P| \geq 3$, this cycle contains both $u_3$ and $v_3$ as vertices.

The subgraph induced by the vertices of the cycle $H = P \cup P^*$ is a sub-graph of a chordal graph, thus it is chordal. Since $v$( = $u_1$) is simplicial, the graph $H - v$ is also chordal. Note that the edge $u_2v_2$ exists (since $u_2$ and $v_2$ are neighbors of $v$), and is an edge of the cycle $v_2 \ldots v_l \ldots u_2$ in $H - v$. Hence, there must be a cycle in $H - v$ with minimal number of vertices containing the edge $u_2v_2$. This cycle must be of size 3, since $H - v$ is chordal. We note the following facts:

\begin{enumerate}
\item This cycle is either of the form $u_2v_qv_2$ or $u_2u_qv_2$, for $q > 2$.

\item If it is of the form $u_2v_qv_2$, then $u_1u_2v_qv_{q+1}\ldots v_l$ is a shortest path from $u_1$ to $v_l = u_m$ (since $u_1v_2v_3 \ldots v_l$ is a shortest path and $q > 2$). This contradicts the fact that $u_2$ must not belong to such a path (since $u_1$ and $u_m$ are discrepancy vertices). Therefore this is an impossibility.

\item If it is of the form $u_2u_qv_2$, and $q > 3$, then there is an edge from $u_2$ to $u_q$. This is a contradiction to the fact that $u_2 u_3 \ldots u_m$ is a shortest path. Therefore, we must have that $q=3$.
\end{enumerate}

Thus we see that this cycle must be precisely the cycle  $u_2u_3v_2$. Therefore, the edge $v_2u_3$ must exist in $G$.

\setcounter{figure}{9}
\begin{figure}[h]
\caption{The induced subgraph $P \cup P^*$}
\centering
\usetikzlibrary{shapes.geometric}

\begin{tikzpicture}
[every node/.style={inner sep=0pt}]
\node (1) [circle, fill] at (87.5pt, -175.0pt) {\textcolor{black}{1}};
\node (2) [circle, fill] at (125.0pt, -150.0pt) {\textcolor{black}{2}};
\node (3) [circle, fill] at (125.0pt, -200.0pt) {\textcolor{black}{3}};
\node (4) [circle, fill] at (175.0pt, -150.0pt) {\textcolor{black}{4}};
\node (5) [circle, fill] at (175.0pt, -200.0pt) {\textcolor{black}{5}};
\node (6) [circle, fill] at (287.5pt, -175.0pt) {\textcolor{black}{6}};
\draw [line width=0.625, color=black] (1) to  (2);
\draw [line width=0.625, color=black] (2) to  (4);
\draw [line width=0.625, dashed, color=black] (4) to  [in=156, out=0] (6);
\draw [line width=0.625, dashed, color=black] (5) to  [in=204, out=0] (6);
\draw [line width=0.625, color=black] (3) to  (5);
\draw [line width=0.625, color=black] (3) to  (1);
\draw [line width=0.625, color=black] (3) to  (2);
\draw [line width=0.625, dotted, color=black] (2) to  (5);
\node at (87.5pt, -160.625pt) {\textcolor{black}{$v_1$ ($u_1$)}};
\node at (125.0pt, -135.625pt) {\textcolor{black}{$v_2$}};
\node at (125.0pt, -214.375pt) {\textcolor{black}{$u_2$}};
\node at (175.0pt, -135.625pt) {\textcolor{black}{$v_3$}};
\node at (175.0pt, -214.375pt) {\textcolor{black}{$u_3$}};
\node at (287.5pt, -189.375pt) {\textcolor{black}{$v_{l}$ ($u_m$)}};
\end{tikzpicture}
\end{figure}
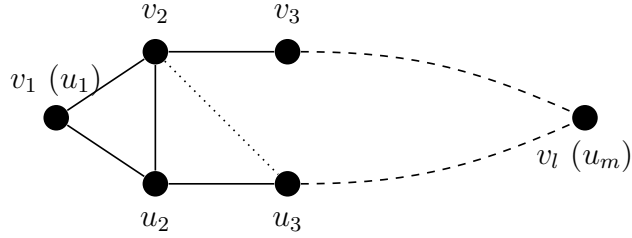

Since $v_2u_3$ exists, we can 2-deform $u_1u_2u_3$ to $u_1v_2u_3$, then 2-deform $v_2u_3 \ldots u_m$ to a shortest path from $v_2$ to $u_m$ (using the inductive assumption), deforming $P$ into an optimal path. Thus we successfully 2-deformed $P$ into an optimal path, and we are done.

If $k \neq m$ we first restrict ourselves to looking at the sub-walk $u_1 \ldots u_k$ of $P$ and apply the argument above to 2-deform it to a shortest path. This has the effect of shortening $P$. If $P$ is not a shortest path as a result of this, we can freely re-apply the same argument a finite number of times to deform $P$ to a shortest path. Specifically, after taking care of $v_1, v_k$ we look for the next pair of discrepancy vertices and apply the argument to them, each time shortening the length of $P$ by at least 1. This can only be done a finite number of times (as $P$ is finite), and at the end of this process we will have deformed $P$ to a shortest path.
\end{proof}

Note that in general chordal graphs are not necessarily planar.

\setcounter{figure}{10}
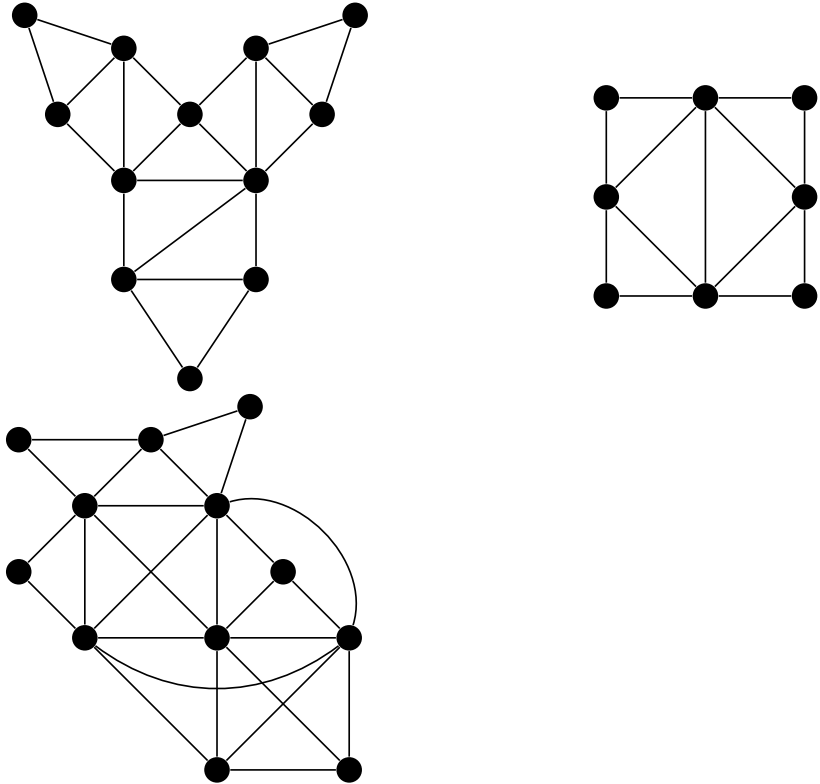
\begin{figure}[!htb]
\caption{The top row shows two maximal outerplanar graphs; the bottom row a non-planar chordal graph}
\begin{subfigure}{.5\textwidth}
\centering
\usetikzlibrary{shapes.geometric}

\begin{tikzpicture}
[every node/.style={inner sep=0pt}]
\node (1) [circle, fill] at (75.0pt, -112.5pt) {\textcolor{black}{1}};
\node (2) [circle, fill] at (87.5pt, -150.0pt) {\textcolor{black}{2}};
\node (3) [circle, fill] at (112.5pt, -125.0pt) {\textcolor{black}{3}};
\node (4) [circle, fill] at (137.5pt, -150.0pt) {\textcolor{black}{4}};
\node (5) [circle, fill] at (112.5pt, -175.0pt) {\textcolor{black}{5}};
\node (6) [circle, fill] at (162.5pt, -175.0pt) {\textcolor{black}{6}};
\node (7) [circle, fill] at (162.5pt, -125.0pt) {\textcolor{black}{7}};
\node (8) [circle, fill] at (187.5pt, -150.0pt) {\textcolor{black}{8}};
\node (9) [circle, fill] at (200.0pt, -112.5pt) {\textcolor{black}{9}};
\node (10) [circle, fill] at (112.5pt, -212.5pt) {\textcolor{black}{1}};
\node (11) [circle, fill] at (162.5pt, -212.5pt) {\textcolor{black}{1}};
\node (12) [circle, fill] at (137.5pt, -250.0pt) {\textcolor{black}{1}};
\draw [line width=0.625, color=black] (1) to  (2);
\draw [line width=0.625, color=black] (1) to  (3);
\draw [line width=0.625, color=black] (2) to  (3);
\draw [line width=0.625, color=black] (3) to  (4);
\draw [line width=0.625, color=black] (2) to  (5);
\draw [line width=0.625, color=black] (5) to  (4);
\draw [line width=0.625, color=black] (4) to  (7);
\draw [line width=0.625, color=black] (7) to  (8);
\draw [line width=0.625, color=black] (6) to  (8);
\draw [line width=0.625, color=black] (4) to  (6);
\draw [line width=0.625, color=black] (7) to  (9);
\draw [line width=0.625, color=black] (8) to  (9);
\draw [line width=0.625, color=black] (5) to  (6);
\draw [line width=0.625, color=black] (10) to  (5);
\draw [line width=0.625, color=black] (11) to  (6);
\draw [line width=0.625, color=black] (10) to  (11);
\draw [line width=0.625, color=black] (10) to  (12);
\draw [line width=0.625, color=black] (12) to  (11);
\draw [line width=0.625, color=black] (6) to  (10);
\draw [line width=0.625, color=black] (6) to  (7);
\draw [line width=0.625, color=black] (5) to  (3);
\end{tikzpicture}
\end{subfigure}%
\begin{subfigure}{.5\textwidth}
\centering
\begin{tikzpicture}
[every node/.style={inner sep=0pt}]
\node (1) [circle, fill] at (75.0pt, -100.0pt) {\textcolor{black}{1}};
\node (2) [circle, fill] at (112.5pt, -100.0pt) {\textcolor{black}{2}};
\node (3) [circle, fill] at (150.0pt, -100.0pt) {\textcolor{black}{3}};
\node (4) [circle, fill] at (150.0pt, -137.5pt) {\textcolor{black}{4}};
\node (5) [circle, fill] at (150.0pt, -175.0pt) {\textcolor{black}{5}};
\node (6) [circle, fill] at (112.5pt, -175.0pt) {\textcolor{black}{6}};
\node (7) [circle, fill] at (75.0pt, -175.0pt) {\textcolor{black}{7}};
\node (8) [circle, fill] at (75.0pt, -137.5pt) {\textcolor{black}{8}};
\draw [line width=0.625, color=black] (6) to  (2);
\draw [line width=0.625, color=black] (1) to  (2);
\draw [line width=0.625, color=black] (2) to  (3);
\draw [line width=0.625, color=black] (3) to  (4);
\draw [line width=0.625, color=black] (4) to  (5);
\draw [line width=0.625, color=black] (5) to  (6);
\draw [line width=0.625, color=black] (6) to  (7);
\draw [line width=0.625, color=black] (7) to  (8);
\draw [line width=0.625, color=black] (8) to  (1);
\draw [line width=0.625, color=black] (2) to  (8);
\draw [line width=0.625, color=black] (4) to  (2);
\draw [line width=0.625, color=black] (6) to  (4);
\draw [line width=0.625, color=black] (6) to  (8);
\end{tikzpicture}
\end{subfigure}
\begin{subfigure}{.5\textwidth}
\centering
\usetikzlibrary{shapes.geometric}

\begin{tikzpicture}
[every node/.style={inner sep=0pt}]
\node (2) [circle, fill] at (125.0pt, -112.5pt) {\textcolor{black}{2}};
\node (12) [circle, fill] at (175.0pt, -162.5pt) {\textcolor{black}{2}};
\node (7) [circle, fill] at (75.0pt, -162.5pt) {\textcolor{black}{7}};
\node (6) [circle, fill] at (125.0pt, -162.5pt) {\textcolor{black}{6}};
\node (8) [circle, fill] at (125.0pt, -212.5pt) {\textcolor{black}{8}};
\node (5) [circle, fill] at (150.0pt, -137.5pt) {\textcolor{black}{5}};
\node (11) [circle, fill] at (175.0pt, -212.5pt) {\textcolor{black}{1}};
\node (3) [circle, fill] at (75.0pt, -112.5pt) {\textcolor{black}{3}};
\node (1) [circle, fill] at (100.0pt, -87.5pt) {\textcolor{black}{1}};
\node (9) [circle, fill] at (50.0pt, -137.5pt) {\textcolor{black}{9}};
\node (13) [circle, fill] at (50.0pt, -87.5pt) {\textcolor{black}{3}};
\node (4) [circle, fill] at (137.5pt, -75.0pt) {\textcolor{black}{4}};
\draw [line width=0.625, color=black] (7) to  (2);
\draw [line width=0.625, color=black] (12) to  [in=17, out=73] (2);
\draw [line width=0.625, color=black] (2) to  (6);
\draw [line width=0.625, color=black] (6) to  (12);
\draw [line width=0.625, color=black] (6) to  (7);
\draw [line width=0.625, color=black] (6) to  (8);
\draw [line width=0.625, color=black] (12) to  (8);
\draw [line width=0.625, color=black] (8) to  (7);
\draw [line width=0.625, color=black] (5) to  (2);
\draw [line width=0.625, color=black] (6) to  (5);
\draw [line width=0.625, color=black] (12) to  (5);
\draw [line width=0.625, color=black] (11) to  (8);
\draw [line width=0.625, color=black] (11) to  (6);
\draw [line width=0.625, color=black] (11) to  (12);
\draw [line width=0.625, color=black] (2) to  (3);
\draw [line width=0.625, color=black] (7) to  (3);
\draw [line width=0.625, color=black] (3) to  (6);
\draw [line width=0.625, color=black] (1) to  (3);
\draw [line width=0.625, color=black] (1) to  (2);
\draw [line width=0.625, color=black] (3) to  (9);
\draw [line width=0.625, color=black] (9) to  (7);
\draw [line width=0.625, color=black] (13) to  (3);
\draw [line width=0.625, color=black] (13) to  (1);
\draw [line width=0.625, color=black] (1) to  (4);
\draw [line width=0.625, color=black] (2) to  (4);
\draw [line width=0.625, color=black] (7) to  [in=217, out=323] (12);
\end{tikzpicture}
\end{subfigure}
\end{figure}

\newpage

\section{The Uniform Limiting Distribution}

In the previous sections we discussed \textit{stable graphs}, graphs for which the pursuit from $s$ to $t$ converges to a unique distribution over all shortest paths. In other cases, while chain pursuit always converges to a distribution over some set of walks, this distribution is not uniquely determined and depends on the initial walk $P(A_0)$ as well as the randomness of the move choices.

The purpose of this section is to prove a general fact about these distributions; namely, that probabilistic chain pursuit will always converge to the \textit{uniform} distribution over a set of walks in one of its closed communicating classes. When restricting the discussion to stable graphs, this says that the pursuit will always converge to the uniform distribution over all shortest paths from $s$ to $t$.

\setcounter{figure}{11}
\begin{figure}[!htb]
\caption{A simulation of the pursuit rule on different graph environments. Each vertex is shaded according to the relative frequency at which an agent $A_i$ was located on it. A corollary of this section is that for a vertex $v$, this frequency converges to precisely the number of walks in $\mathcal{C}$ (the closed communicating class at which the pursuit stabilized) that pass through $v$, divided by $|\mathcal{C}|$. In the graphs pictured, the most frequented vertices are those closest to the straight line from $s$ to $t$.}
\centering
\begin{subfigure}{.5\textwidth}
\centering
\includegraphics[width=\linewidth]{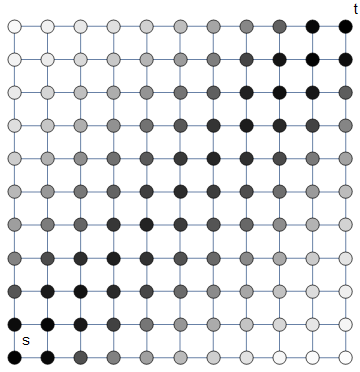}
\end{subfigure}%
\begin{subfigure}{.5\textwidth}
\centering
\includegraphics[width=\textwidth]{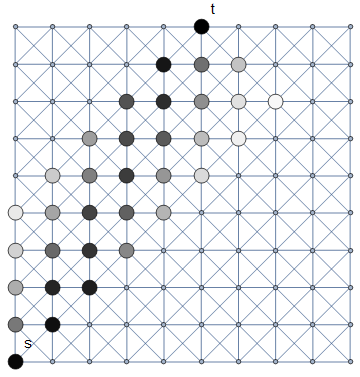}
\end{subfigure}
\label{fig:graphdistribution}
\end{figure}

\begin{proposition}
\label{uniformdistribution}
Let $\mathcal{C}$ be a closed communicating class of $\mathcal{M}_{\Delta} (s,t)$. The limiting distribution of $\mathcal{M}$ restricted to $\mathcal{C}$ is the uniform distribution.
\end{proposition}

In particular we have:

\begin{proposition}
Let $G$ be a stable graph. Then the limiting distribution of $\mathcal{M}_{\Delta} (s,t)$ (for any choice of $\Delta$) is the uniform distribution over all shortest paths.
\end{proposition}

We note that for the purposes of our proof, unlike the previous sections, we cannot restrict the parameter $\Delta$ here to the value `$2$', and our argument must hold for any value $\Delta$ greater than $1$. This is because the transition probabilities of $\mathcal{M}_{\Delta} (s,t)$ depend on the choice of $\Delta$.

Let $\mathcal{C}$ be a closed communicating class of $\mathcal{M}_{\Delta} (s,t)$, and consider the Markov chain defined by deleting every walk not in $\mathcal{C}$ from $\mathcal{M}$. This restricted Markov chain is finite and aperiodic (as every walk taken by an ant has a non-zero probability of immediately recurring for the next ant) and so has a limiting distribution. Let $U = u_1 \ldots u_n$ be an arbitrary, fixed walk in $\mathcal{C}$. We will assume that $P(A_0)$ is uniformly distributed over all paths in $\mathcal{C}$ and show that the induced distribution of $P(A_1)$ is the uniform distribution. This is equivalent to showing that the uniform distribution is the limiting distribution in the restricted Markov chain.

\theoremstyle{definition}
\begin{definition}
\label{shuffles}
Let $X = x_1 \ldots x_n \in \mathcal{C}$ be a walk in $G$. An \textit{(i,j)-shuffle} of $X$, written ${S}^i_j(X)$, is a random variable resulting from the replacement of the sub-walk $x_i \ldots x_j$ of $X$ with a shortest path from $x_i$ to $x_j$ chosen uniformly at random from all such paths. (We define $S^i_i(X) = X$ and $S^i_{n+k}(X) = S^i_n(X)$ for $k \geq 1$).
\end{definition}

Consider the pursuit of $A_0$ by the ant $A_1$. The pursuit rule states that $A_1$ moves, at time $i$, to a vertex determined by choosing at random one of the shortest paths from $A_1$ to $A_0$ and stepping on the first vertex of that path. As we will see, this is equivalent to sequentially performing $(i,i+\Delta)$-shuffles of $P(A_0)$ starting from $i=1$ up to $i=n$.

We define $\eta(u,v)$ to be the number of shortest paths from $u$ to $v$ in $G$, and we define $\eta(v,v) = 1$. In a slight abuse of notation, we will also let $\eta(u,v)$ be the \textit{set} of all such shortest paths, trusting that the intent will be clear from the context. We further define $\eta(u_1u_2,v)$ to be the set of all shortest paths from $u_1$ to $v$, whose first edge is $u_1u_2$ (i.e. paths with $d(u_1, v)$ edges whose first edge is $u_1u_2$). We note that $\eta(u_1u_2,v)$ might equal $0$ for some choices of $u_1u_2$ and $v$. 

For simplicity, we define $x_{n+k} = x_n$ for $k \geq 1$, in all applications below.

Write $p_U(X)$ for the probability $Prob[P(A_1) = X|P(A_0) = U]$. We start with the following lemma:

\begin{lemma}
\label{probabilitytransitionformula}
Let $X = x_1 \ldots x_n$ be a walk in $\mathcal{C}$ such that $p_U(X) > 0$. Then we have: $$p_U(X) = \prod_{i=1}^{n} \frac{\eta(x_ix_{i+1},u_{i+\Delta})}{\eta(x_i,u_{i+\Delta})}$$
\end{lemma}

\begin{proof}
In order for $A_1$ to have $P(A_{1}) = X$, it must, at time $i+1$, choose the vertex $x_{i+1}$, having already chosen the vertex $x_i$ at time $i$. At this time $A_0$ will be on the vertex $u_{i+\Delta}$, that is at distance $\Delta$ from $x_i$ by $\Delta$-optimality. By the pursuit rule, it follows that moving to the vertex $x_{i+1}$ happens with probability $$\frac{\eta(x_ix_{i+1},u_{i+\Delta})}{\eta(x_i,u_{i+\Delta})}$$  The formula follows from multiplication of all these probabilities.
\end{proof}

We define the following stochastic process: let $U^1 = S^1_{1+\Delta}(U)$, and $U^i = S^i_{i+\Delta}(U^{i-1})$. Define $\widetilde{p}_U(X)$ to be the probability that $U^n$ will equal $X$. We will show:

\begin{lemma}
\label{shuffleprobabilityequivalence}
For all $X \in \mathcal{C}$, $p_U(X) = \widetilde{p}_U(X)$.
\end{lemma}
\begin{proof}
The $i$th shuffle permanently determines the $(i+1)$th vertex. We have $u_1 = x_1$ and $u_n = x_n$. In order for $u_2$ to be changed into $x_2$ we must have that the second vertex of $S^1_{1 + \Delta}(U)$ is $x_2$, which by the definition of $\eta$ happens with probability $$\frac{\eta(x_1x_{2},u_{1+\Delta})}{\eta(x_1,u_{1+\Delta})}$$

Inductively, we again arrive at the formula:

 $$\widetilde{p}_U(X) = \prod_{i=1}^{n} \frac{\eta(x_ix_{i+1},u_{i+\Delta})}{\eta(x_i,u_{i+\Delta})}$$
 
 Which shows that $p_U(X) = \widetilde{p}_U(X)$. 
\end{proof}

From the proof of \ref{shuffleprobabilityequivalence} we see that the probability distribution of $U^n$ is equivalent to that of $P(A_1)$ (conditioned on $P(A_0) = U$).

\begin{lemma}
\label{shuffleuniformlemma}
Let $U = u_1 \ldots u_n$ be chosen uniformly at random from the walks in $\mathcal{C}$. For all $X = x_1 \ldots x_n \in \mathcal{C}$, $Prob[U^n = X] = \frac{1}{|\mathcal{C}|}$.
\end{lemma}
\begin{proof}
According to our assumption, for any $X \in \mathcal{C}$, we have $Prob[U = X] = \frac{1}{|\mathcal{C}|}$. Then for any $1 \leq i < j \leq i+\Delta$ we have that $$Prob[S^i_j(U) = X] = \frac{\eta(u_i,u_j)}{|\mathcal{C}|} \cdot \frac{1}{\eta(u_i,u_j)} = \frac{1}{|\mathcal{C}|}$$ (The computation relies on the fact that the subwalk from $u_i$ to $u_j$ in $U$ must be $\Delta$-deformable to any path in $\eta(u_i,u_j)$, thus by closure, the result of any such deformation must be in $\mathcal{C}$). 

We see that the distribution of $S^i_j(U)$ is the same as $U$. Since $U^n$ is just a composition of a finite number of such shuffles, we have that $Prob[U^n = X] = \frac{1}{|\mathcal{C}|}$.
\end{proof}

The above lemma shows that the distribution of $P(A_1)$ is the uniform distribution, completing the proof of Proposition \ref{uniformdistribution}.

\section{Conclusion and Future Work}
\label{chap:conclusion}

We studied the behavior of ``ants'' that form an idealized ``ant trail'' through probabilistic chain pursuits of each other. Our pursuit rule is a natural generalization of the pursuit rule of \cite{ants2} to arbitrary graph environments. The behavior of this more general rule as time tends to infinity is not necessarily ``nice'' in all graph environments--unlike the original simpler scenario where pursuit was restricted to a grid graph. Over an arbitrary graph environment, convergence to a shortest path is not guaranteed, and the pursuit is not guaranteed to always stabilize in the same way. 

We investigated conditions under which convergence and stability do occur. In doing so we extended the results of \cite{ants2} on the grid in multiple ways--by showing several classes of graphs that are convergent to shortest paths and are stable, and by showing that the limiting distribution of walks in probabilistic pursuit is always uniform. From the graph classes investigated in section 3, two of these, pseudo-modular graphs and graph products, include the special case where the underlying graph is a grid.

Three immediate extensions of this work can readily be considered. The first is a classification of stable and convergent graphs with respect to a parameter $\Delta$ greater than 2, i.e., graphs where convergence to the shortest path (or, additionally, to a unique limiting distribution) is guaranteed only for choice of $\Delta$ greater than 2. The second is further investigation of types of planar graphs that are stable or convergent. The third is an analysis of the computational complexity of the problem of deciding whether a given graph is convergent or stable.

\bibliographystyle{plain}
\bibliography{general.bib}







\end{document}